\documentclass[american,aps,prl,reprint,superscriptaddress]{revtex4-2}
\usepackage[T1]{fontenc}
\usepackage[utf8]{inputenc}
\usepackage{xcolor}
\usepackage{babel}
\usepackage{physics}
\usepackage{amsmath}
\usepackage{amsthm}
\usepackage{amssymb}
\usepackage[unicode=true,pdfusetitle,
 bookmarks=true,bookmarksnumbered=false,bookmarksopen=false,
 breaklinks=false,pdfborder={0 0 0},pdfborderstyle={},backref=false,colorlinks=true]
 {hyperref}
\hypersetup{
 allcolors=magenta}

\makeatletter
\theoremstyle{plain}
\newtheorem{thm}{\protect\theoremname}
\theoremstyle{plain}
\newtheorem{prop}[thm]{\protect\propositionname}

\usepackage{times}
\usepackage{txfonts}
\usepackage{braket}
\usepackage{colortbl}
\usepackage{float}
\usepackage{graphicx}

\makeatother

\providecommand{\theoremname}{Theorem}
\providecommand{\propositionname}{Proposition}

\begin{document}
\title{Second Law of Entanglement Manipulation with Entanglement Battery}

\author{Ray Ganardi}
\email{ray@ganardi.xyz}
\affiliation{Centre for Quantum Optical Technologies, Centre of New Technologies, University of Warsaw, Banacha 2c, 02-097 Warsaw, Poland}
\affiliation{School of Physical and Mathematical Sciences, Nanyang Technological University, 21 Nanyang Link, Singapore, 637371
}

\author{Tulja Varun Kondra}
\affiliation{Institute for Theoretical Physics III, Heinrich Heine University D\"{u}sseldorf, Universit\"{a}tsstra{\ss}e 1, D-40225 D\"{u}sseldorf, Germany}

\author{Nelly H.Y. Ng}
\affiliation{School of Physical and Mathematical Sciences, Nanyang Technological University, 21 Nanyang Link, Singapore, 637371
}
\affiliation{Centre for Quantum Technologies, National University of Singapore, 3 Science Drive 2, 117543, Singapore}

\author{Alexander Streltsov}
\affiliation{Institute of Fundamental Technological Research, Polish Academy of Sciences, Pawińskiego 5B, 02-106 Warsaw, Poland}
\affiliation{Centre for Quantum Optical Technologies, Centre of New Technologies, University of Warsaw, Banacha 2c, 02-097 Warsaw, Poland}

\begin{abstract}
A central question since the beginning of quantum information science is how two distant parties can convert one entangled state into another.
It has been conjectured that such conversions could be executed reversibly in an asymptotic regime, mirroring the reversible nature of Carnot cycles in classical thermodynamics.
While a conclusive proof of this conjecture has been missing so far, earlier studies have excluded reversible entanglement manipulation in various settings.
In this work, we show that arbitrary mixed state entanglement transformations can be made reversible under local operations and classical communication, when assisted by an entanglement battery--an auxiliary quantum system that stores and supplies entanglement in a way that ensures no net entanglement is lost.
In particular, the rate of transformation in the asymptotic limit can be quantitatively expressed as a ratio of entanglement present within the quantum states involved.
Our setting allows to consider different entanglement quantifiers which give rise to unique principles governing state transformations, effectively constituting diverse manifestations of a ``second law'' of entanglement manipulation.
These findings resolve a long-standing open question on the reversible manipulation of entangled states and are also applicable to multipartite entanglement and other quantum resource theories, including quantum thermodynamics.
\end{abstract}

\maketitle

Over the past decades, striking parallels between the principles governing manipulation of entangled systems and the laws of thermodynamics~\cite{PhysRevLett.89.240403,Brandao2008,Brandao2010,Brandao2010b,Lami2023} have been revealed. A prime illustration of this similarity is often cast with the narrative of two agents, Alice and Bob, sharing \( n \) copies of an entangled state \( \ket{\psi} \) that they can also manipulate. It is known that under certain conditions, by utilizing simple operations such as local operations and classical communication (LOCC)~\cite{RevModPhys.81.865,PhysRevA.53.2046}, Alice and Bob can transform their initially shared state into \( m \) copies of another state \( \ket{\phi} \). In the asymptotic limit of large \( m, n \), we can have $m \approx n$ if and only if the entanglement entropy $E(\psi) = S(\psi_A)$ does not increase~\cite{PhysRevA.53.2046}.
Here, $S(\rho) = -\Tr\pqty{\rho \log \rho}$ denotes the von Neumann entropy.
This rule mirrors a fundamental concept in classical thermodynamics, where the entropy of a system uniquely determines its potential for interconversion through adiabatic processes~\cite{LIEB19991}.

The striking similarity between entanglement theory and thermodynamics naturally leads to the question: Does there exist a ``second law of entanglement manipulation''~\cite{PhysRevLett.89.240403}, akin to its thermodynamic counterpart?
In other words, is there an ``entropy-like'' quantity that governs all state transformations of entangled systems?
This question is conceptually tied to the feasibility of reversibility, specifically in the sense of zero loss rate in the asymptotic limit. 
Let us denote the maximal rate of transformation of $n$ copies of $\rho$ to $m$ copies of $\sigma$ as $R(\rho \to \sigma) = \sup \Bqty{m/n}$, where the supremum is taken over a set of all allowed protocols of interest~\cite[Section 11.5]{gour2024resourcesquantumworld}.
When two states $\rho, \sigma$ can be transformed reversibly, then $R(\rho \to \sigma) R(\sigma \to \rho) = 1$.
Analogously to Carnot's theorem, reversibility is connected to the efficiency of conversion between any two states.
To see this, let us define the loss rate in a cycle $\rho \to \sigma \to \rho$ as $\ell(\rho \to \sigma \to \rho) = 1 - R(\rho \to \sigma) R(\sigma \to \rho)$.
Then, a given protocol converts two states losslessly if and only if it is reversible.

It is known that for any two pure states $\psi, \phi$, we have $R(\psi \to \phi) = S(\psi_A)/S(\phi_A)$~\cite{PhysRevA.53.2046}, and in particular $R(\psi \to \phi) R(\phi \to \psi) = 1$.
In other words, the transformation is reversible.
This implies that the distillation rate of singlets determines \emph{all} transformation rates $R(\psi \to \phi) = R(\psi \to \psi^-)/R(\phi \to \psi^-)$, giving us a ``second law of entanglement manipulation'' for pure states.
However, the situation no longer holds for mixed states.
Extreme examples are the well-studied bound entangled states~\cite{PhysRevLett.80.5239}, where entanglement is required for their formation (under LOCC), yet one cannot extract any usable entanglement from these states in the form of singlets~\cite{PhysRevLett.80.5239,vidal2001irreversibility}.

To obtain a second law for entanglement beyond pure states, one has to move beyond the traditional LOCC paradigm. Over the past decade, a prominent approach was to connect the problem of state transformations to that of hypothesis testing~\cite{Brandao2008}.
In this framework, all protocols that inject a small amount of entanglement into the system were allowed, provided that this contribution vanishes in the asymptotic limit.
Nonetheless, it is important to note that the assertions in~\cite{Brandao2008} rely on a landmark result known as the generalized quantum Stein's lemma~\cite{Brandao2010}. Recent scrutiny~\cite{Berta2023gapinproofof} has raised questions about the validity of this lemma, casting uncertainty on the robustness of the conclusions drawn from its application.

Ref.~\cite{Lami2023} motivated recent renewed interest in the problem of (ir)reversibility. It was shown that when restricting to operations that preserve the set of separable states—known as non-entangling operations—then entanglement manipulation is necessarily irreversible.
Even when we boost LOCC with catalysis, reversibility cannot be obtained due to the existence of bound entanglement in this setting~\cite{lami2023catalysis,ganardi2023catalytic}.
One can in principle recover reversibility by allowing the use of non-CPTP maps~\cite{chen2025reversible}.
However, these maps are not physical and any implementation must involve classical postprocessing that comes with an associated sampling overhead.
A different approach shows reversibility by adopting probabilistic protocols~\cite{regula2024reversibility}, which necessarily changes the meaning of reversibility.
Indeed, in this we are only guaranteed that there is a non-zero chance that we recover the initial state.
Are these sacrifices necessary to obtain a reversible theory of entanglement?

In this article, we contribute to the above ongoing debate, by demonstrating the existence of a physically meaningful framework in which entanglement reversibility can be achieved. We introduce the concept of entanglement batteries~\cite{alhambra2019entanglement} into our framework. An entanglement battery is an additional entangled system shared between Alice and Bob. In order to prevent the embezzling of resources~\cite{PhysRevA.67.060302}, the battery must be returned at the end of the procedure with at least the same amount of entanglement. This idea can be viewed as a generalization of entanglement catalysis~\cite{PhysRevLett.83.3566}, a topic that has garnered significant interest recently~\cite{Datta_2023,lipka2023catalysis}. An entangled catalyst, in this context, is an ancillary quantum system in an entangled state provided to Alice and Bob. They are allowed to employ this catalyst in their transformation process, with the requirement that it must be returned in its original state. Our approach extends this setup by allowing changes in the state of the ancillary system, provided that there is no reduction in its entanglement.

\begin{figure*}
    \centering
    \includegraphics[width=0.7\paperwidth]{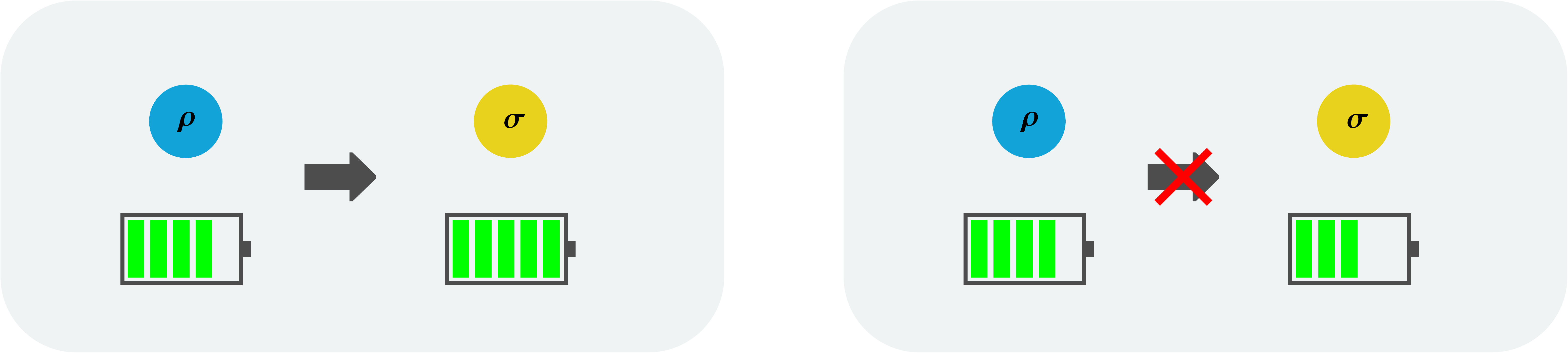}
    \caption{\textbf{State transformations with an entanglement battery.} A conversion from $\rho$ to $\sigma$ is possible if $\rho \otimes \tau$ can be converted into $\sigma \otimes \tilde \tau$ via LOCC, and the amount of entanglement in the battery does not decrease (left part of the figure). Transformations which decrease the entanglement in the battery are not allowed (right part of the figure).}
    \label{fig:LOCCEB}
\end{figure*}

\textbf{\emph{Reversible entanglement manipulations with entanglement battery.}}
We demonstrate that reversible manipulations across all entangled states become feasible when an entanglement battery is incorporated.
Specifically, we consider a setting where Alice and Bob can perform LOCC, and additionally are given access to a supplementary shared entangled state, which is our \emph{entanglement battery}.
This setting has been previously studied in Ref.~\cite{alhambra2019entanglement}, where it was shown that we can formulate a fluctuation theorem in pure state entanglement transformations with an explicit battery model. In this article, we generalize this framework to allow for a general battery and we study transformations between general mixed states.
We show that we obtain a reversible theory when we impose the following restriction on the battery: Alice and Bob must not decrease the level of entanglement within the battery throughout the process. The setup is shown in Fig.~\ref{fig:LOCCEB}. 

More specifically, we examine transformations of the form 
\begin{equation}
    \rho \otimes \tau \rightarrow \sigma \otimes \tilde{\tau},
\end{equation}
where \( \tau \) represents the initial state of the battery, and \( \tilde{\tau} \) denotes its final state. The requirement that the battery does not lose any entanglement then amounts to
\begin{equation}
    E(\tilde{\tau}) \geq E(\tau), \label{eq:Battery}
\end{equation}
where the choice of the entanglement quantifier \(E\) plays a pivotal role in our analysis. We show that varying the choice of entanglement quantifiers leads to the emergence of distinct ``second laws'' of entanglement manipulation. 

The basic property of any entanglement measure $E$ is that it should not increase under LOCC, i.e., $E(\Lambda[\rho])\leq E(\rho)$ for any LOCC protocol $\Lambda$~\cite{PhysRevLett.78.2275}. A particularly interesting class of entanglement measures for our analysis are quantifiers which are additive, i.e., they satisfy $E(\rho \otimes \tau) = E(\rho) + E(\tau)$ for any states $\rho$ and $\tau$. Additionally, we will consider entanglement quantifiers that exhibit asymptotic continuity~\cite{Synak-Radtke_2006}. This property implies that for two states close to each other in trace distance, the difference in the amount of entanglement can grow at most logarithmically in the Hilbert space dimension.
We further require the entanglement measure to be finite for all states. The subsequent theorem offers a comprehensive characterization of state transformations within this setup.
\begin{thm} \label{thm:1}
A state $\rho$ can be converted into another state $\sigma$ via LOCC with an entanglement battery if and only if 
\begin{equation}
E(\rho) \geq E(\sigma). \label{eq:Main}
\end{equation}
Here, $E$ is an additive and finite entanglement measure.
\end{thm}

We provide a concise outline of our proof, with comprehensive details available in the End Matters. First, we leverage the properties of additive entanglement measures to establish that under the framework of LOCC augmented with an entanglement battery, the amount of entanglement does not increase. Subsequently, we introduce a protocol that transforms the state \( \rho \) into \( \sigma \) for any pair of states that comply with the stipulations outlined in Eq.~(\ref{eq:Main}). This protocol is remarkably straightforward: the battery is initialized in the target state \( \sigma \), and then local transformations are employed to interchange the states of the main system and the battery. A similar technique has been used recently in \cite{lipka2023catalysis} in the context of correlated catalysis, see also~\cite{alhambra2019entanglement}.

Let us now consider the multi-copy setting, where Alice and Bob are provided with \( n \) copies of the initial quantum state \( \rho \). The aim is to produce \( m \) copies of the target state \( \sigma \). We are interested in the maximal rate of the transformation $m/n$, allowing for an error which vanishes in the asymptotic limit. The following theorem gives a complete characterization of the transformation rates if Alice and Bob perform an LOCC protocol with an entanglement battery.

\begin{thm} \label{thm:2}
    The maximal conversion rate for converting $\rho$ into $\sigma$ via LOCC with an entanglement battery is given by
    \begin{equation}
R(\rho\rightarrow\sigma)=\frac{E(\rho)}{E(\sigma)}.
\end{equation}
Here, $E$ is an entanglement measure which is finite, additive and asymptotically continuous.
\end{thm}
\noindent An outline of the proof is provided here, with detailed elaboration available in the End Matters. Employing the properties of additive and asymptotically continuous entanglement measures, we first demonstrate that the transformation rate is upper-bounded by $E(\rho)/E(\sigma)$. The converse is shown by introducing a protocol that accomplishes the transformation at the aforementioned rate. Echoing the approach of Theorem~\ref{thm:1}, the optimal protocol involves a battery in an entangled state, specifically comprising copies of the target state \( \sigma \). The transformation \( \rho \rightarrow \sigma \) at the rate of $E(\rho)/E(\sigma)$ can be realized by permuting the main system with the battery. 

An immediate consequence of our work is the construction of a framework of reversible entanglement manipulation. More specifically, we establish that $R(\rho \rightarrow \sigma) R(\sigma \rightarrow \rho) = 1$ for any $\rho$ and $\sigma$, implying that in the asymptotic limit, any two entangled states \( \rho \) and \( \sigma \) can be interconverted reversibly. Furthermore, when the ratio $E(\rho)/E(\sigma)$ is rational, reversible interconversion is feasible even with a finite number of copies. We further notice that the optimal conversion rate can be achieved through a protocol involving local operations alone, which interestingly obviates the need for classical communication in this task. We suspect that this is an artifact of allowing the battery to contain any finite amount of entanglement. However, we anticipate that achieving the conversion with minimal entanglement in the battery may require some level of classical communication.
In addition, we can argue that this framework is tight with respect to reversibility;
in the Supplemental Material we show that if we require that the entanglement in the battery does not decrease when measured by two non-equivalent measures, then we cannot obtain a reversible framework.

An example of an entanglement quantifier that satisfies our criteria, namely being additive and asymptotically continuous, is the squashed entanglement, defined as~\cite{10.1063/1.1643788}
\begin{equation}
    E(\rho^{AB}) = \text{inf}\Bqty{\frac{1}{2}I(A;B|E): \rho^{ABE}\,\,\text{extension of}\,\, \rho^{AB} }, \label{eq:SquashedEntanglement}
\end{equation}
with the quantum conditional mutual information $I(A; B|E)$. We refer to the Supplemental Material for more details about the properties of squashed entanglement. We note that the squashed entanglement is not the only entanglement quantifier having these properties, another example is given in the Supplemental Material. 

Our findings not only affirm the existence of a second law for entangled state transformations, but also link it to the squashed entanglement present in the involved quantum states. Prior to our work, it was hypothesized~\cite{Brandao2008,PhysRevLett.89.240403} that should reversible manipulation of entangled states be achievable in any framework, the rates at which such reversible transformations occur would be linked to a different measure of entanglement: the regularized relative entropy of entanglement. This quantity is defined as~\cite{PhysRevLett.78.2275} \begin{equation} \label{eq:Eregularized}
 E_{\mathrm{r}}^{\infty}(\rho)=\lim_{n\rightarrow\infty}\frac{1}{n}E_{\mathrm{r}}(\rho^{\otimes n})
\end{equation}
with the relative entropy of entanglement defined as $E_{\mathrm r}(\rho)=\min_{\sigma\in\mathcal{S}}S(\rho||\sigma)$, and the quantum relative entropy $S(\rho||\sigma)=\mathrm{Tr}(\rho\log_{2}\rho)-\mathrm{Tr}(\rho\log_{2}\sigma)$. Here, the minimization is done over the set of separable (i.e. non-entangled) states $\mathcal S$.

 As detailed in the Supplementary Materials, we can extend our methodologies to establish a version of the second law of entanglement manipulations predicated on the relative entropy of entanglement.
Additionally, we illustrate a phenomenon termed \emph{self-dilution}: the asymptotic conversion of $n$ copies of an entangled state into $m > n$ copies of itself, with an error that can be made arbitrarily small in the asymptotic limit. This is possible when we use an additive but not asymptotically continuous entanglement measure such as logarithmic negativity. More details with concrete examples are discussed in the Supplemental Material. 

Furthermore, the findings discussed thus far are readily adaptable to multipartite scenarios. In situations involving $N$ parties, the objective becomes transforming an $N$-partite state $\rho$ into another $N$-partite state $\sigma$, utilizing an ancillary system that may also exhibit entanglement across all $N$ parties. In such contexts, Theorems~\ref{thm:1} and \ref{thm:2} remain applicable when the entanglement measure employed is the multipartite squashed entanglement~\cite{5075874}, since the latter is additive and asymptotically continuous. This implies that in this framework, reversible transitions are feasible between any multipartite entangled states.

We will now discuss the difference of our methods from the approach considered in Ref.~\cite{Lami2023}. Specifically, Ref.~\cite{Lami2023} posits that reversible manipulations of entangled states are unattainable in frameworks reliant on non-entangling operations. In such frameworks, Alice and Bob, possessing many copies of a state \( \rho \), are restricted to quantum operations that cannot generate entanglement from separable states. Under these conditions, the authors of~\cite{Lami2023} demonstrate the impossibility of reversible interconversion between certain entangled states \( \rho \) and \( \sigma \), i.e., \( R(\rho \rightarrow \sigma)  R(\sigma \rightarrow \rho) < 1 \). It is important to note that the methods in Ref.~\cite{Lami2023} do not cover the transformations considered in this work, because we allow the agents to have access to an additional system.
This opens a degree of freedom since Alice and Bob can further optimize the protocol (i.e. operation \emph{and} battery) in a way that depends on the system's initial and final states.
This is critical in enabling reversible interconversions within our framework, and moreover is a standard feature of catalytic processes~\cite{PhysRevLett.83.3566,Datta_2023,lipka2023catalysis}. Additionally, it is crucial to clarify that in our setup, Alice and Bob are also precluded from creating entanglement from separable states if the entanglement measure $E$ is additive and vanishes solely on separable states. These conditions are satisfied, for example, by the squashed entanglement~\cite{Li2014,10.1063/1.1643788,Alicki_2004}. In summary, we are focusing on non-entangling state transformations instead of non-entangling operations.

In the preceding discussion, we operated under the assumption that the main system and the battery return to an uncorrelated state at the procedure's conclusion. However, this requirement can be relaxed to allow for correlations between them, with the constraint that the amount of entanglement in the battery does not decrease. In this scenario, any entanglement measure $E$ yields a reversible theory provided it satisfies four conditions: monotonicity under LOCC, continuity, additivity on product states, and general superadditivity, expressed as $E(\rho^{AA'BB'}) \geq E(\rho^{AB})+E(\rho^{A'B'})$. Superadditivity has been shown to be a critical feature of resource monotones in general~\cite{PhysRevLett.126.150502,takagi2022correlation,lipka2023catalysis}. 
Notably, for pure bipartite states, transformations in this framework are entirely governed by the entanglement entropy, regardless of the entanglement measure chosen.
This mirrors the uniqueness of entanglement entropy as the measure governing the transformations under standard LOCC~\cite{donald2002uniqueness,PhysRevA.53.2046,PhysRevLett.127.150503}.
For an in-depth discussion, we refer to the Supplemental Material.

Not all entanglement measures prove useful within this context. Specifically, for certain measures such as geometric entanglement, the constraint in Eq.~(\ref{eq:Battery}) fails to impose limitations on potential transformations, permitting arbitrary amplification of entanglement within the main system. This phenomenon bears resemblance to the concept of entanglement embezzlement, previously discussed in~\cite{PhysRevA.67.060302}.
We refer the interested reader to the Supplemental Material for more details.

\textbf{\emph{Quantum thermodynamics.}}
The framework described in this work has immediate implications beyond entanglement theory.
The second law of thermodynamics says that state transformations in classical thermodynamics are governed by the free energy~\cite{LIEB19991}.
However, this statement relies on several assumptions such as the negligibility of fluctuations and energetic coherence, which might not necessarily hold in the quantum regime. In particular, one commonly refers to quantum states without any coherence in the energy eigenbasis as energy-incoherent states. 

A commonly used model to study the thermodynamics of quantum systems is the framework of thermal operations, which explicitly models interactions with a thermal bath through an energy-preserving unitary~\cite{janzing2000thermodynamic}.
For energy-incoherent states, the transformations are fully governed by a family of inequalities called the thermomajorization conditions~\cite{horodecki2013fundamental}. The original form of the second law is subsequently recovered, if we consider catalytic transformations that allows correlations between the system and catalyst~\cite{PhysRevX.8.041051}.
More precisely, it was shown that for any two energy-incoherent states $\rho, \sigma$, there exists a thermal operation that transforms $\rho^S$ into $\sigma^{S'}$ if and only if $F(\rho^S, H^S) \geq F(\sigma^{S'}, H^{S'})$, where $F$ is the free energy, defined as
\begin{align}
    F(\rho, H) = k_B T \pqty{S(\rho \| \gamma) - \log Z}.
\end{align}
Here, $T$ is the temperature $Z$ is the partition function, and $\gamma = \exp \pqty{-H/k_B T}/Z$ is the thermal state associated with the system (that is characterized by $H$). Numerical evidences even suggest that low-dimensional catalysts already provide significant advantages~\cite{son2022catalysis}, when it comes to the problem of simplifying the required unitary control over system and bath.

The obvious enhancement of catalysis in quantum thermodynamics even for energy-incoherent state transformations raises the following question: is it possible to extend the second law to coherent states?
Previous results suggests that coherence and catalysis can interact in unexpected ways~\cite{son2023hierarchy,tajima2024gibbspreserving,lie2023catalysis,Kondra_2024,Shiraishi_2024}, and thus it is unclear if such an extension is possible.
%
We answer this question by showing that the free energy determines general state transformations under thermal operations when we allow access to a thermodynamic battery.
Formally, we study transformations of the form $\rho^S \otimes \tau^B \to \sigma^{S'} \otimes \tilde{\tau}^{B'}$, along with the requirement that the free energy does not decrease $F(\tilde{\tau}^{B'}, H^{B'}) \geq F(\tau^B, H^B)$. 
Intuitively, this allows the battery to act as a source of coherence without simultaneously giving away free energy.
By analogous arguments to Theorem~\ref{thm:1}, we can show that at any finite temperature $T < \infty$, the allowed set of transformations is governed by free energy, i.e.\ $\rho^S \to \sigma^{S'}$ if and only if $F(\rho^S, H^S) \geq F(\sigma^{S'}, H^{S'})$.

\begin{thm}\label{thm:TO}
    A state $\pqty{\rho^S, H^S}$ can be converted into another state $\pqty{\sigma^{S'}, H^{S'}}$ with thermal operations and a free energy battery if and only if $F(\rho^S, H^S)\geq F(\sigma^{S'}, H^{S'})$.
\end{thm}

We emphasize that Theorem~\ref{thm:TO} holds for generic states, as opposed to most results in single-shot thermodynamics that hold only for energy-incoherent states.
%
It is worth noting that the state transformation conditions in Theorem~\ref{thm:TO} are identical to that of Gibbs-preserving operations with correlated catalysis~\cite{PhysRevLett.126.150502}.
Thus, our theorem provides an operational interpretation of catalytic Gibbs-preserving operations as thermal operations augmented with a free energy battery.


The choice of free energy $F$ in Theorem~\ref{thm:TO} is not unique: in the quantum regime, there exists a family of generalized free energies $F_{\alpha}$~\cite{brandao2015second} that determines the allowed transformations with exact catalysis.
All of these generalized free energies are additive, and therefore using them in our framework will lead to transformations that are governed by a single monotone.
We can go further and relate the resource change in the battery to that in the system:
\begin{align}\label{eq:resource-balance}
    f(\rho^S, H^S) - f(\sigma^{S'}, H^{S'}) \geq f(\tilde{\tau}^{B'}, H^{B'}) - f(\tau^B, H^B),
\end{align}
which hold for any additive monotone $f$.
However, when we allow correlations, then Eq.~(\ref{eq:resource-balance}) is equivalent to the local monotonicity property studied in Ref.~\cite{Boes_2022}.
There, it was shown that the standard free energy $F$ is essentially the only function that is locally monotonic and continuous, up to additive and multiplicative constants.
Combined with our results, this gives a formulation of the unique thermodynamic second law for coherent quantum systems.

\textbf{\emph{Conclusions.}}
In this article, we have explored the concept of entanglement batteries and their impact on the manipulation of entangled systems. Our findings illuminate the path toward achieving reversible entanglement manipulations across all quantum states, thereby addressing a long-standing challenge in quantum information science. 

Our results lead to a family of ``second laws'' for entanglement manipulation, each characterized by the specific measure used to quantify entanglement of the system. Notably, we demonstrate that, for certain entanglement quantifiers, the asymptotic conversion rates for any two states $\rho$ and $\sigma$ take a particularly simple form $E(\rho)/E(\sigma)$.

This work opens several avenues for future research, presenting questions pivotal for a deeper understanding of entangled systems. A particularly compelling area for further investigation involves identifying all entanglement quantifiers that yield conversion rates in the form of $E(\rho)/E(\sigma)$. It was previously speculated~\cite{Brandao2008} that the regularized relative entropy of entanglement might uniquely possess this characteristic, but our findings hint at the possibility that other entanglement quantifiers may also be suitable for this task.

Our techniques are not limited to the domain of entanglement but are applicable to a broad spectrum of quantum resources~\cite{RevModPhys.91.025001}. This can be achieved through generalizing the concept of entanglement battery to a \emph{resource battery} -- a supplementary system that participates in the transformation process without a decrease of the resource in question. Although the principle of reversibility has been confirmed in various quantum resource theories by other methods, our framework stands out as a comprehensive model. It can systematically enable the demonstration of reversibility across quantum resource theories based on a minimal set of assumptions. 

\emph{Note added.} After the completion of our work, two proofs of the generalized quantum Stein's lemma were posted on arXiv~\cite{hayashi2024generalizedquantumsteinslemma,lami2024solutiongeneralisedquantumsteins}, establishing the reversibility of entanglement theory under asymptotically non-entangling operations.
These results are complementary to our work, and show that reversibility of entanglement can be obtained in multiple frameworks.

This work was supported by the National Science Centre Poland (Grant No. 2022/46/E/ST2/00115) and within the QuantERA II Programme (Grant No. 2021/03/Y/ST2/00178, acronym ExTRaQT) that has received funding from the European Union's Horizon 2020 research and innovation programme under Grant Agreement No. 101017733. The work of TVK is supported by the German Federal Ministry of Education and Research (BMBF) within the funding program “quantum technologies – from basic research to market” in the joint project QSolid (grant number 13N16163). RG and NHYN are supported by the start-up grant for Nanyang Assistant Professorship of Nanyang Technological University, Singapore.


\section*{End Matters}

\subsection{LOCC with entanglement battery}
Let us start with a formal definition of the procedure considered in our article. We say that $\rho^{AB}$ can be converted into $\sigma^{AB}$ via LOCC with entanglement battery if there exists an LOCC protocol $\Lambda$ and states $\tau^{A'B'}$ and $\tilde\tau^{A'B'}$ such that
\begin{equation}
\Lambda(\rho^{AB}\otimes\tau^{A'B'})=\sigma^{AB}\otimes\tilde{\tau}^{A'B'}. \label{eq:LOCCEB-1}
\end{equation}
For more details about LOCC protocols and their features we refer to Ref.~\cite{Chitambar2014}. In the following, we denote $AB$ as the \emph{main system}, and $A'B'$ comprises the entanglement battery.
Moreover, we require that the final state of the battery $\tilde \tau$ has at least the same amount of entanglement as the initial state $\tau$, i.e., 
\begin{equation}
E(\tilde{\tau}^{A'B'}) \geq E(\tau^{A'B'}). \label{eq:LOCCEB-2}
\end{equation}
If $E$ is continuous, then without loss of generality we can even assume that the entanglement of the battery is conserved, as we can always mix the final state of the battery with a separable state to decrease the final battery entanglement. A similar scenario has been introduced in~\cite{alhambra2019entanglement}, without the entanglement non-decreasing condition on the battery.

In the asymptotic setting, we say that $\rho$ can be converted into $\sigma$ with an achievable rate $r$ via LOCC with entanglement battery, if for any $\varepsilon,\delta > 0$, there are integers $m, n$, an LOCC protocol $\Lambda$, and a battery state $\tau$ such that
\begin{subequations} \label{eq:Asymptotic}
\begin{align}
\Lambda\left(\rho^{\otimes n}\otimes\tau^{C}\right) & =\mu^{S_{1}\ldots S_{m}}\otimes\tilde{\tau}^{C},\\
\left\Vert \mu^{S_{1}\ldots S_{m}}-\sigma^{\otimes m}\right\Vert _{1} & <\varepsilon,\\
E(\tilde{\tau}^{C}) & \geq E(\tau^{C}),\label{eq:Asymptotic-E}\\
\frac{m}{n}& >r-\delta .
\end{align}
\end{subequations}
Here, each system $S_i$ denotes a copy of the bipartite system $AB$, and $C$ denotes the battery system, which is also bipartite.
Furthermore, entanglement in the battery is measured by a fixed measure $E$.
The supremum over all achievable rates $r$ is denoted by $R(\rho \rightarrow \sigma)$. 

In the setting defined above, the battery is not correlated with the main system at the end of the procedure. In the Supplemental Material, we discuss the more general setting where correlations between the main system and the battery are taken into account.

\subsection{Proof of Theorems 1 and 2 of the main text}

We begin with the assumption that $E$ is an additive and asymptotically continuous measure~\cite{Synak-Radtke_2006}, i.e., 
\begin{align}
E^{AA'|BB'}(\rho^{AB}\otimes\sigma^{A'B'}) & =E^{A|B}(\rho^{AB})+E^{A'|B'}(\sigma^{A'B'}),\\
\left|E(\rho)-E(\sigma)\right| & \leq K\left\Vert \rho-\sigma\right\Vert _{1}\log_{2}d+f(\left\Vert \rho-\sigma\right\Vert _{1}).
\end{align}
Here, $K > 0$ is a constant, $d$ is the dimension of the Hilbert space, and $f(x)$ is some function which does not depend on $d$ and vanishes in the limit $x \rightarrow 0$.

To prove Theorem~1 of the main text, let us first assume that $\rho^{AB}$ can be converted into $\sigma^{AB}$ via LOCC with entanglement battery. Then, there is an LOCC protocol $\Lambda$ and states $\tau^{A'B'}$ and $\tilde\tau^{A'B'}$ such that Eqs.~(\ref{eq:LOCCEB-1}) and ~(\ref{eq:LOCCEB-2}) are fulfilled. Using the additivity of $E$, finiteness and its monotonicity under LOCC, we obtain
\begin{align}
E(\sigma^{AB}) & =E(\sigma^{AB}\otimes\tilde{\tau}^{A'B'})-E(\tilde{\tau}^{A'B'})\\
 & \leq E(\rho^{AB}\otimes\tau^{A'B'})-E(\tilde{\tau}^{A'B'}) \leq E(\rho^{AB}).\nonumber 
\end{align}
This shows that the amount of entanglement in the main system $AB$ cannot increase in this procedure.

To prove the converse, let $\rho$ and $\sigma$ be two states fulfilling 
\begin{equation}
E(\rho) \geq E(\sigma). \label{eq:Proof-1}
\end{equation}
A conversion $\rho \rightarrow \sigma$ can be achieved by choosing
\begin{equation}
\tau^{A'B'}=\sigma^{A'B'},
\end{equation}
and the LOCC protocol consists of local permutations of $A$ and $A'$ on Alice's side, and correspondingly $B$ and $B'$ on Bob's side. Performing this protocol, the overall initial state $\rho^{AB} \otimes \sigma^{A'B'}$ is converted into $\sigma^{AB}\otimes \rho^{A'B'}$. Thus, the state of the battery at the end of the process is given by
\begin{equation}
\tilde{\tau}^{A'B'}=\rho^{A'B'}.
\end{equation}
Due to Eq.~(\ref{eq:Proof-1}) we have $E(\tilde \tau) \geq E(\tau)$ as required, and thus this protocol achieves the transformation $\rho^{AB} \rightarrow \sigma^{AB}$.
This completes the proof of Theorem~1 of the main text.

Theorem~1 further implies that a state $\rho\otimes\mu$ can be converted into another state $\sigma \otimes \mu'$ via LOCC with entanglement battery if and only if $E(\rho)-E(\sigma)\geq E(\mu')-E(\mu)$. Here, the system in the state $\mu$ can be considered as an additional component of a battery. While $E(\rho)-E(\sigma) > 0$, i.e. the main system loses entanglement during the procedure, the battery is capable of storing entanglement to compensate for this loss. 

We will now prove Theorem~2 of the main text, showing that for any entanglement measure which is additive and asymptotically continuous, the asymptotic transformation rates take the form 
\begin{equation}
R(\rho\rightarrow\sigma)=\frac{E(\rho)}{E(\sigma)}.
\end{equation}
For this, we will first show that the rate is upper bounded by $E(\rho)/E(\sigma)$. From Eqs.~(\ref{eq:Asymptotic}), additivity of $E$, and the fact that $E$ is nonincreasing under LOCC, it follows that 
\begin{align}
E\left(\mu^{S_{1}\ldots S_{m}}\right)+E\left(\tilde{\tau}^{C}\right) & =E\left(\mu^{S_{1}\ldots S_{m}}\otimes\tilde{\tau}^{C}\right)\leq E\left(\rho^{\otimes n}\otimes\tau^{C}\right)\nonumber \\
 & =nE(\rho)+E(\tau),
\end{align}
which together with Eq.~(\ref{eq:Asymptotic-E}) implies 
\begin{equation}
E\left(\mu^{S_{1}\ldots S_{m}}\right)\leq nE(\rho).
\end{equation}
Using again Eqs.~(\ref{eq:Asymptotic}) with asymptotic continuity of $E$ we arrive at
\begin{equation}
E\left(\sigma^{\otimes m}\right)\leq E\left(\mu^{S_{1}\ldots S_{m}}\right)+K\varepsilon m\log_{2}d+f(\varepsilon),
\end{equation}
where $d$ is the dimension of $S_i$. Combining these results we obtain
\begin{equation}
mE(\sigma)=E\left(\sigma^{\otimes m}\right)\leq nE(\rho)+K\varepsilon m\log_{2}d+f(\varepsilon),
\end{equation}
which can also be expressed as
\begin{equation}
\frac{m}{n}\leq\frac{E(\rho) + \frac{f(\varepsilon)}{n}}{E(\sigma) - \varepsilon K \log_2 d},
\end{equation}
when
\begin{equation}
0<\varepsilon<\frac{E(\sigma)}{K\log_{2}d} \label{eq:EpsilonProof}.
\end{equation}
Finally, using Eqs.~(\ref{eq:Asymptotic}), we see that any feasible rate $r$ must fulfill
\begin{equation}
r<\frac{E(\rho) + \frac{f(\varepsilon)}{n}}{E(\sigma) - \varepsilon K \log_2 d} + \delta.
\end{equation}
Recalling that we can choose arbitrary $\delta > 0$ and $\varepsilon$ in the range~(\ref{eq:EpsilonProof}), it follows that the asymptotic transformation rate is upper bounded by $E(\rho)/E(\sigma)$, as claimed.

We will now present a protocol achieving conversion at rate $E(\rho)/E(\sigma)$. Assume first that $E(\rho)/E(\sigma)$ is a rational number, i.e., there exist integers $m$ and $n$ such that 
\begin{equation}
\frac{m}{n}=\frac{E(\rho)}{E(\sigma)}. \label{eq:Rational}
\end{equation}
In this case, we can choose the initial state of the battery to be $\tau = \sigma^{\otimes m}$, and the total initial state is $\rho^{\otimes n} \otimes \sigma^{\otimes m}$~\footnote{
If $\rho^{\otimes n}$ and $\sigma^{\otimes m}$ have different dimensions,
product states are appended accordingly, making the dimensions of
the main system and the battery equal.}. 
Similar to the proof of Theorem~1, Alice and Bob now apply local permutations, permuting the main system and the battery. The final state is given by $\sigma^{\otimes m} \otimes \rho^{\otimes n}$, where the battery is now in the state $\tilde \tau = \rho^{\otimes n}$. The final amount of entanglement of the battery is given by $E(\tilde \tau) = nE(\rho)$, which is equal to the initial amount of entanglement $E(\tau) = mE(\sigma)$ due to Eq.~(\ref{eq:Rational}). This proves that for rational $E(\rho)/E(\sigma)$ a transformation with this rate is achievable.

If $E(\rho)/E(\sigma)$ is irrational, then for any $\varepsilon > 0$ there are integers $m$ and $n$ such that 
\begin{equation}
\frac{E(\rho)}{E(\sigma)}-\varepsilon<\frac{m}{n}<\frac{E(\rho)} {E(\sigma)}. \label{eq:Irrational}
\end{equation}
Alice and Bob now use the same procedure as in the rational case. The amount of entanglement in the battery does not decrease in this procedure, since $nE(\rho) > mE(\sigma)$ due to Eq.~(\ref{eq:Irrational}). Since $\varepsilon > 0$ in Eq.~(\ref{eq:Irrational}) can be chosen arbitrarily, Alice and Bob can also achieve conversion at rate $E(\rho)/E(\sigma)$ in this case.









\appendix


\section*{Supplemental Material}

\subsection{Multiple measures on the battery}
Let us explore what happens when we impose two non-equivalent~\cite{Virmani_2000} entanglement measures on the battery.
By non-equivalent, we mean that we can find two states $\rho, \sigma$ such that $E_1 (\rho) > E_1(\sigma)$, but $E_2(\rho) < E_2(\sigma)$.
By following the arguments in Theorem~2, we can show that the transformation rates are bounded as follows
\begin{align}
  R(\rho \to \sigma) &\leq \min\Bqty{
                       \frac{E_1(\rho)}{E_1(\sigma)},
                       \frac{E_2(\rho)}{E_2(\sigma)}
                       }
  = \frac{E_2(\rho)}{E_2(\sigma)} < 1,\\
  R(\sigma \to \rho) &\leq \min\Bqty{
                       \frac{E_1(\sigma)}{E_1(\rho)},
                       \frac{E_2(\sigma)}{E_2(\rho)}
                       }
                       = \frac{E_1(\sigma)}{E_1(\rho)} < 1,
\end{align}
so that $R(\rho \to \sigma) R(\sigma \to \rho) < 1$.
This shows that the battery-assisted transformations only becomes reversible when we impose the non-decreasing condition on only one entanglement measure.

\subsection{Proof of Theorem~3 of the main text}
We know that the free energy $F$ is an additive monotone for thermal operations.
Furthermore, when $0 < T < \infty$, then $F$ is always finite.
Therefore, we can repeat the arguments of Theorem~1 to obtain the only if direction.

For the if direction, we simply note that swapping the system and battery along with their Hamiltonians is allowed in thermal operations.
Then, we can run the protocol in the proof of Theorem~1 to show the if direction.

\subsection{Zero error conversion}
We will now consider a more restricted version of the conversion problem, which we term \emph{zero error conversion}. We say that $\rho$ can be converted into $\sigma$ with zero error via LOCC with entanglement battery if for any $\delta > 0$, there exist integers $m$, $n$, an LOCC protocol $\Lambda$, and states of the battery $\tau$ and $\tilde \tau$ such that
\begin{subequations} \label{eq:ZeroError}
\begin{align}
\Lambda\left(\rho^{\otimes n}\otimes\tau^{C}\right) & =\sigma^{\otimes m}\otimes\tilde\tau^{C},\\
E(\tilde\tau^{C}) & \geq E(\tau^{C}),\\
\frac{m}{n}+\delta & >r.
\end{align}
\end{subequations}
In comparison to Eqs.~(\ref*{eq:Asymptotic}), no error is allowed in the final state, i.e., the state at the end of the procedure should be exactly $\sigma^{\otimes m}\otimes\tilde\tau^{C}$. The supremum over all feasible rates $r$ in this process will be called $R_{\mathrm{ze}}(\rho \rightarrow \sigma)$. 

For zero-error conversion, Theorem~2 is true for any additive entanglement measure, i.e., asymptotic continuity is not required, and it holds that
\begin{equation}
R_{\mathrm{ze}}(\rho\rightarrow\sigma)=\frac{E(\rho)}{E(\sigma)}.
\end{equation}
To see this, we can write 
\begin{align}
mE(\sigma)+E(\tilde{\tau}) & =E\left(\sigma^{\otimes m}\otimes\tilde{\tau}\right)\leq E\left(\rho^{\otimes n}\otimes\tau\right)\\
 & =nE(\rho)+E(\tau),\nonumber 
\end{align}
which implies that 
\begin{equation}
\frac{m}{n}\leq\frac{E(\rho)}{E(\sigma)}.
\end{equation}
This means that the maximal transformation rate is upper bounded by $E(\rho)/E(\sigma)$. The converse can be seen with the same protocol as in the proof of Theorem~2, see the discussion below Eq.~(\ref*{eq:Rational}).

\subsection{Entanglement measures which are additive and asymptotically continuous}
As mentioned in the main text, an example for an entanglement measure which is additive and asymptotically continuous is the squashed entanglement given in Eq.~(5) of the main text, with the quantum conditional mutual information of a state $\rho^{ABE}$ defined as
\begin{equation}
I(A;B|E)=S(\rho^{AE})+S(\rho^{BE})-S(\rho^{ABE})-S(\rho^{E}).
\end{equation}
Additivity of the squashed entanglement on all states has been proven in~\cite{10.1063/1.1643788}, while asymptotic continuity is a direct consequence of the results presented in~\cite{Alicki_2004}. Moreover, squashed entanglement is zero on separable states, and is larger than zero otherwise~\cite{Li2014}.

Another example is the conditional entanglement of mutual information, defined as~\cite{PhysRevLett.101.140501,yang2007conditional} 
\begin{equation}
E(\rho^{AB})=\inf\frac{1}{2}\left\{ I\left(AA':BB'\right)-I\left(A':B'\right)\right\} ,
\end{equation}
where the infimum is taken over all extensions $\rho^{AA'BB'}$ of the state $\rho^{AB}$, and $I(X:Y)$ denotes the quantum mutual information of the state $\rho^{XY}$:
\begin{equation}
I(X:Y)=S(\rho^{X})+S(\rho^{Y})-S(\rho^{XY}).
\end{equation}
Additivity and asymptotic continuity of this entanglement measure has been proven in~\cite{PhysRevLett.101.140501,yang2007conditional}. Moreover, this measures is zero on separable states and larger than zero otherwise, as follows from the fact that it vanishes on separable states~\cite{PhysRevLett.101.140501} and is lower bounded by squashed entanglement~\cite{yang2007conditional}.

\subsection{Other entanglement measures}

Many entanglement measures known in the literature are not generally additive or do not satisfy asymptotic continuity. However, we demonstrate that the presented approach can be applied to some commonly used measures, even if these properties are not fulfilled, or not known to hold. An important example is the relative entropy of entanglement~\cite{PhysRevLett.78.2275}, defined as
\begin{equation}
    E_{\mathrm r}(\rho)=\min_{\sigma\in\mathcal{S}}S(\rho||\sigma),
\end{equation}
where $\mathcal S$ is the set of separable states, and $S(\rho||\sigma)=\mathrm{Tr}(\rho\log_{2}\rho)-\mathrm{Tr}(\rho\log_{2}\sigma)$ is the quantum relative entropy.
Since the relative entropy of entanglement is not additive~\cite{PhysRevA.64.062307}, our Theorems do not directly apply in this case.
Nevertheless, as we show below, quantifying the amount of entanglement in the battery by the relative entropy of entanglement leads to a theory with bounded distillation rates, in particular the asymptotic rate for converting a state $\rho$ into $\ket{\psi^-}$ is upper bounded by $E_\mathrm{r}(\rho)$.
Moreover, we show that $E_{\mathrm{r}}^{\infty}(\rho)/E_{\mathrm{r}}^{\infty}(\sigma)$ is a feasible transformation rate for any two entangled states $\rho$ and $\sigma$, where $E_\mathrm{r}^\infty$ is the regularized relative entropy of entanglement defined in Eq.~(6) of the main text.
The same rate is also feasible if the amount of entanglement in the battery is quantified via $E_{\mathrm{r}}^{\infty}$.
In this setting, we also obtain a theory with bounded distillation rates. We note that for some states $\rho$ and $\sigma$ the asymptotic transformation rate might exceed $E_{\mathrm{r}}^{\infty}(\rho)/E_{\mathrm{r}}^{\infty}(\sigma)$. 

Another frequently used entanglement measure in the literature is the logarithmic negativity~\cite{PhysRevA.58.883,PhysRevA.65.032314,PhysRevLett.95.090503}, defined as
\begin{equation}
E_{\mathrm{n}}(\rho)=\log_{2}\left\Vert \rho^{T_{A}}\right\Vert _{1},
\end{equation}
where $T_A$ denotes partial transposition. Although this measure is additive, it does not satisfy asymptotic continuity~\cite{Lami2023}. As discussed above in the Supplemental Material, the additivity of $E_\mathrm{n}$ implies a reversible framework for zero error transformations. Below, we demonstrate that the entanglement distillation rate is bounded even when allowing for an error margin in the transformation. Specifically, the asymptotic rate for the conversion $\rho \rightarrow \ket{\psi^-}$ is given by $E_\mathrm{n}(\rho)$. We further note that logarithmic negativity vanishes if and only if the state has positive partial transpose (PPT), which means that this framework allows for the creation of PPT states from separable states.

We will now show that measuring entanglement with logarithmic negativity entails a curious phenomenon, where it is possible to dilute a state $\rho$ into more copies of itself.
More precisely, we demonstrate an asymptotic conversion of $n$ copies of a quantum state $\rho$ into a noisy state which is close to $m/n$ copies of the initial state $\rho$, where $m > n$.
Note that the error may vanishes only in the limit $n \rightarrow \infty$.
We call this phenomenon \emph{self-dilution}. We emphasize that this effect does not lead to embezzling of entanglement, since the entanglement distillation rates are finite for any given state.

\begin{figure}
\centering
\includegraphics[width=0.9\columnwidth]{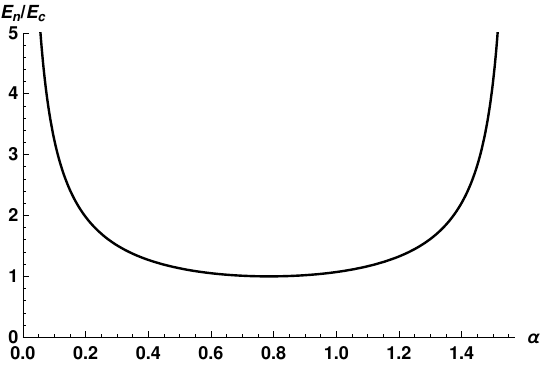}

\caption{\label{fig:Self-Dilution}Self-dilution rate $E_{\mathrm{n}}(\ket{\psi})/E_{\mathrm{c}}(\ket{\psi})$ for $\ket{\psi}=\cos\alpha\ket{00}+\sin\alpha\ket{11}$ as a function of $\alpha$. }

\end{figure}

First, recall that we can convert any state $\rho$ into singlets at rate $E_\mathrm n(\rho)$ in this setup.
As discussed above in the Supplemental Material, this conversion can be achieved with zero error.
Moreover, by using LOCC it is possible to convert singlets approximately into the state $\rho$ at rate $1/E_{\mathrm{c}}(\rho)$, where $E_{\mathrm c}$ is the entanglement cost~\cite{Plenioquant-ph/0504163}.
Thus, performing these operations in sequence will convert $\rho^{\otimes n}$ into a state which is close to $nE_{\mathrm{n}}(\rho)/E_{\mathrm{c}}(\rho)$ copies of $\rho$, and the error can be made arbitrarily small in the limit of large $n$.
Self-dilution of $\rho$ occurs whenever $E_{\mathrm{n}}(\rho)/E_{\mathrm{c}}(\rho)>1$. 

In Fig.~\ref{fig:Self-Dilution} we show the rate $E_{\mathrm{n}}(\ket{\psi})/E_{\mathrm{c}}(\ket{\psi})$ as a function of $\alpha$ for the states  $\ket{\psi}=\cos\alpha\ket{00}+\sin\alpha\ket{11}$. The rate is above one as long as $\ket{\psi}$ is not maximally entangled, and diverges in the limit $\alpha \rightarrow 0$.

This phenomenon is not unique to logarithmic negativity, or even entanglement theory.
By analogous arguments, we can show that self-dilution occurs in thermodynamics for incoherent states if we use $F_{\textrm{max}} (\rho) = \inf \Bqty{\log \lambda \,|\, \rho \leq \lambda \gamma}$ to quantify the resources in the battery.
Similarly to logarithmic negativity, $F_{\textrm{max}} (\rho)$ is not asymptotically continuous and has the operational interpretation as the exact cost of preparing $\rho$ from many copies of $\ketbra{1}$~\cite{horodecki2013fundamental}.
This will be discussed in more detail below in the Supplemental Material

While counter intuitive, self-dilution does not mean that there are no cost associated with creating more copies of $\rho$.
This is because it is not clear that we can repeat the protocol to obtain even more copies of $\rho$, and in fact this is forbidden since the distillation rate to the singlet state is bounded.
Physically, the battery is providing the extra entanglement/work that is needed to create more copies of $\rho$.
Therefore, this scenario is relevant when only certain types of resource in the battery are scarce~\cite{Ng_2017}, even in the asymptotic limit.

\subsubsection{Relative entropy of entanglement}
Let us discuss the consequences of constraining the amount of entanglement in the battery by the relative entropy of entanglement defined as~\cite{PhysRevLett.78.2275}
\begin{equation}
E_{\mathrm r}(\rho)=\min_{\sigma\in\mathcal{S}}S(\rho||\sigma)
\end{equation}
with the quantum relative entropy
\begin{equation}
S(\rho||\sigma)=\mathrm{Tr}(\rho\log_{2}\rho)-\mathrm{Tr}(\rho\log_{2}\sigma),
\end{equation}
and $\mathcal S$ denotes the set of separable states. 
In particular, we are interested in asymptotic transformations with rates defined in Eqs.~(\ref*{eq:Asymptotic}). This means that condition~(\ref*{eq:Asymptotic-E}) is replaced by 
\begin{equation}
E_{\mathrm{r}}(\tilde{\tau})\geq E_{\mathrm{r}}(\tau). \label{eq:RelativeEntropyProof}
\end{equation}
While the non-additivity of relative entropy of entanglement prevents a straightforward application of our Theorems, we can show that this framework gives rise to a nontrivial resource theory, with bounded entanglement distillation rates. For this, we recall some useful properties of the relative entropy of entanglement. In particular, $E_{\mathrm r}$ does not increase under LOCC~\cite{PhysRevLett.78.2275} and is subadditive, i.e., for any two states $\rho$ and $\tau$ it holds that
\begin{equation}
E_{\mathrm{r}}^{AA'|BB'}(\rho^{AB}\otimes\tau^{A'B'})\leq E_{\mathrm{r}}^{A|B}(\rho^{AB})+E_{\mathrm{r}}^{A'|B'}(\tau^{A'B'}).
\end{equation}
In addition, the inequality becomes an equality if (at least) one of the states is pure~\cite{rubboli2023new}.

Additionally, we will use the following continuity property which was proven implicitly in~\cite{Winter2016}.
For completeness, we will provide a proof.

\begin{prop} \label{prop:ER}
For any two states $\rho^{AB}$ and $\sigma^{AB}$ with $\frac{1}{2}||\rho^{AB}-\sigma^{AB}||_{1}\leq\varepsilon$
and any $\tau^{A'B'}$, it holds 
\begin{align}
 & \left|E_{\mathrm{r}}^{AA'|BB'}\left(\rho^{AB}\otimes\tau^{A'B'}\right)-E_{\mathrm{r}}^{AA'|BB'}\left(\sigma^{AB}\otimes\tau^{A'B'}\right)\right|\nonumber \\
 & \leq\varepsilon\log_{2}d_{AB}+(1+\varepsilon)h\left(\frac{\varepsilon}{1+\varepsilon}\right),
\end{align}
where $h(x)=-x\log_{2}x-(1-x)\log_{2}(1-x)$ is the binary entropy and $d_{AB}$ is the dimension of $AB$.
\end{prop}
\begin{proof}
The proof is implicitly contained in the proofs of Lemma 2 and Lemma
7 in Ref.~\cite{Winter2016}. Without loss of generality we can set $||\rho^{AB}-\sigma^{AB}||_{1}=2\varepsilon$.
We now define the state
\begin{equation}
\Delta=\frac{1}{\varepsilon}(\rho-\sigma)_{+},
\end{equation}
where $(\rho-\sigma)_{+}$ denotes the positive part of $\rho-\sigma$.
Using the same arguments as in the proof of Lemma 2 in Ref.~\cite{Winter2016}, we
find
\begin{align}
\rho\otimes\tau & =\sigma\otimes\tau+(\rho-\sigma)\otimes\tau\nonumber \\
 & \leq\sigma\otimes\tau+\varepsilon\Delta\otimes\tau\nonumber \\
 & =(1+\varepsilon)\omega\otimes\tau, \label{eq:ProofEr-1}
\end{align}
where we have defined the state 
\begin{equation}
\omega=\frac{1}{1+\varepsilon}\sigma+\frac{\varepsilon}{1+\varepsilon}\Delta.
\end{equation}
Due to Eq.~(\ref{eq:ProofEr-1}), we can define another state 
\begin{equation}
\Delta'=\frac{1+\varepsilon}{\varepsilon}\omega-\frac{1}{\varepsilon}\rho,
\end{equation}
which implies that $\omega$ can also be written as
\begin{equation}
\omega=\frac{1}{1+\varepsilon}\rho+\frac{\varepsilon}{1+\varepsilon}\Delta'.
\end{equation}
Using convexity of the relative entropy of entanglement~\cite{PhysRevA.57.1619}, we find
\begin{equation}
E_{\mathrm{r}}(\omega\otimes\tau)\leq\frac{1}{1+\varepsilon}E_{\mathrm{r}}(\sigma\otimes\tau)+\frac{\varepsilon}{1+\varepsilon}E_{\mathrm{r}}(\Delta\otimes\tau).\label{eq:ProofEr-2}
\end{equation}

Now, let $\gamma\in\mathcal{S}$ be a separable state such that $E_{\mathrm{r}}(\omega\otimes\tau)=S(\omega\otimes\tau||\gamma)$.
Recalling that the von Neumann entropy fulfills $S(\sum_{i}p_{i}\rho_{i})\leq\sum_{i}p_{i}S(\rho_{i})+H(p)$ with $H(p)=-\sum_{i}p_{i}\log_{2}p_{i}$, we obtain 
\begin{align}
E_{\mathrm{r}}(\omega\otimes\tau) & =-S(\omega\otimes\tau)-\mathrm{Tr}(\omega\otimes\tau\log_{2}\gamma)\label{eq:ProofEr-3}\\
 & \geq-h\left(\frac{\varepsilon}{1+\varepsilon}\right)-\frac{1}{1+\varepsilon}S(\rho\otimes\tau)-\frac{\varepsilon}{1+\varepsilon}S(\Delta'\otimes\tau)\nonumber \\
 & -\frac{1}{1+\varepsilon}\mathrm{Tr}(\rho\otimes\tau\log_{2}\gamma)-\frac{\varepsilon}{1+\varepsilon}\mathrm{Tr}(\Delta'\otimes\tau\log_{2}\gamma)\nonumber \\
 & =-h\left(\frac{\varepsilon}{1+\varepsilon}\right)+\frac{1}{1+\varepsilon}S\left(\rho\otimes\tau||\gamma\right)+\frac{\varepsilon}{1+\varepsilon}S\left(\Delta'\otimes\tau||\gamma\right)\nonumber \\
 & \geq-h\left(\frac{\varepsilon}{1+\varepsilon}\right)+\frac{1}{1+\varepsilon}E_{\mathrm{r}}(\rho\otimes\tau)+\frac{\varepsilon}{1+\varepsilon}E_{\mathrm{r}}(\Delta'\otimes\tau).\nonumber 
\end{align}
Using Eqs.~(\ref{eq:ProofEr-2}) and (\ref{eq:ProofEr-3}), we arrive
at the inequality
\begin{align}
E_{\mathrm{r}}(\rho\otimes\tau)-E_{\mathrm{r}}(\sigma\otimes\tau) & \leq\varepsilon\left[E_{\mathrm{r}}(\Delta\otimes\tau)-E_{\mathrm{r}}(\Delta'\otimes\tau)\right]\nonumber \\
 & +(1+\varepsilon)h\left(\frac{\varepsilon}{1+\varepsilon}\right).
\end{align}

Therefore, to complete the proof of the proposition, it is enough to show that 
\begin{equation}
E_{\mathrm{r}}(\Delta\otimes\tau)-E_{\mathrm{r}}(\Delta'\otimes\tau)\leq\log_{2}d_{AB}.\label{eq:ProofEr-4}
\end{equation}
For this, note that 
\begin{equation}
E_{\mathrm{r}}(\Delta\otimes\tau)\leq E_{\mathrm{r}}(\Phi_{d_{AB}}\otimes\tau)=\log_{2}d_{AB}+E_{\mathrm{r}}(\tau),
\end{equation}
where $\ket{\Phi_{d_{AB}}}=\sum_{i=0}^{d_{AB}-1}\ket{ii}/\sqrt{d_{AB}}$
is a maximally entangled state on $AB$. Using $E_{\mathrm{r}}(\Delta'\otimes\tau)\geq E_{\mathrm{r}}(\tau)$,
we arrive at Eq.~(\ref{eq:ProofEr-4}), and the proof is complete. 
\end{proof}

Now, consider an asymptotic conversion $\rho\rightarrow\phi^{+}$ with $\ket{\phi^{+}}=(\ket{00}+\ket{11})/\sqrt{2}$. Using Eqs.~(\ref*{eq:Asymptotic}) together with the properties of $E_{\mathrm r}$ mentioned above, we obtain
\begin{align} \label{eq:ProofEr-5}
m+E_{\mathrm{r}}\left(\tilde{\tau}\right) & =E_{\mathrm{r}}\left(\ket{\phi^{+}}\!\bra{\phi^{+}}^{\otimes m}\otimes\tilde{\tau}\right)\\
 & \leq E_{\mathrm{r}}\left(\mu^{S_{1}\ldots S_{m}}\otimes\tilde{\tau}\right)+\frac{\varepsilon}{2}m+\left[1+\frac{\varepsilon}{2}\right]h\left(\frac{\frac{\varepsilon}{2}}{1+\frac{\varepsilon}{2}}\right)\nonumber\nonumber \\
 & \leq E_{\mathrm{r}}\left(\rho^{\otimes n}\otimes\tau\right)+\frac{\varepsilon}{2}m+\left[1+\frac{\varepsilon}{2}\right]h\left(\frac{\frac{\varepsilon}{2}}{1+\frac{\varepsilon}{2}}\right)\nonumber\nonumber \\
 & \leq nE_{\mathrm{r}}\left(\rho\right)+E_{\mathrm{r}}\left(\tau\right)+\frac{\varepsilon}{2}m+\left[1+\frac{\varepsilon}{2}\right]h\left(\frac{\frac{\varepsilon}{2}}{1+\frac{\varepsilon}{2}}\right).\nonumber\nonumber 
\end{align}
This expression is equivalent to 
\begin{equation}
\frac{m}{n}\leq\frac{1}{1-\frac{\varepsilon}{2}}\left[E_{\mathrm{r}}\left(\rho\right)\!-\!\frac{1}{n}\left[E_{\mathrm{r}}\left(\tilde{\tau}\right)\!-\!E_{\mathrm{r}}\left(\tau\right)\right]+\frac{1}{n}\left[1\!+\!\frac{\varepsilon}{2}\right]h\left(\frac{\frac{\varepsilon}{2}}{1+\frac{\varepsilon}{2}}\right)\right].
\end{equation}
Using Eq.~(\ref{eq:RelativeEntropyProof}) we further obtain 
\begin{equation}
\frac{m}{n}\leq\frac{1}{1-\frac{\varepsilon}{2}}\left[E_{\mathrm{r}}\left(\rho\right)+\frac{1}{n}\left[1+\frac{\varepsilon}{2}\right]h\left(\frac{\frac{\varepsilon}{2}}{1+\frac{\varepsilon}{2}}\right)\right].
\end{equation}
Recalling that we can choose arbitrary $\varepsilon > 0$, these results imply that the asymptotic rate for converting $\rho$ into $\phi^+$ in this framework is bounded above by $E_{\mathrm{r}}(\rho)$, i.e., 
\begin{equation}\label{distil}
R(\rho\rightarrow\phi^{+})\leq E_{\mathrm{r}}(\rho).
\end{equation}

For any two entangled states $\rho$ and $\sigma$, we will now show that $E_{\mathrm{r}}^{\infty}(\rho)/E_{\mathrm{r}}^{\infty}(\sigma)$ is a feasible rate for the conversion $\rho \rightarrow \sigma$,
where $E_{\mathrm{r}}^{\infty}$ is the regularized relative entropy of entanglement
\begin{equation}
E_{\mathrm{r}}^{\infty}(\rho)=\lim_{n\rightarrow\infty}\frac{1}{n}E_{\mathrm{r}}(\rho^{\otimes n}). \label{eq:Ereg}
\end{equation}
For proving this, consider first two states $\rho$ and $\sigma$ such that
\begin{equation}
    E_{\mathrm{r}}^{\infty}(\rho)>E_{\mathrm{r}}^{\infty}(\sigma). \label{eq:ProofEreg}
\end{equation}
Then, it must be that $E_{\mathrm{r}}(\rho^{\otimes n})/n>E_{\mathrm{r}}(\sigma^{\otimes n})/n$ for all large enough $n$, and thus also 
\begin{equation}
    E_{\mathrm{r}}(\rho^{\otimes n})>E_{\mathrm{r}}(\sigma^{\otimes n}). \label{eq:ProofEr}
\end{equation}
It is now possible to convert $\rho^{\otimes n}$ into $\sigma^{\otimes n}$ using a similar protocol as in the proof of Theorem~2. For this, we choose the initial state of the battery to be $\tau = \sigma^{\otimes n}$, i.e., the total initial state is $\rho^{\otimes n}\otimes\sigma^{\otimes n}$. Permuting the main system and the battery (which can be achieved via local operations only), we obtain the final state $\sigma^{\otimes n}\otimes\rho^{\otimes n}$, where the battery is now in the state $\tilde{\tau}=\rho^{\otimes n}$. Due to Eq.~(\ref{eq:ProofEr}), it holds that $E(\tilde{\tau})>E(\tau)$, i.e., the amount of entanglement in the battery does not decrease. 

The above arguments show that for any two states fulfilling Eq.~(\ref{eq:ProofEreg}) the asymptotic rate for converting $\rho$ into $\sigma$ is at least one. Consider now two general states $\rho$ and $\sigma$, not necessarily fulfilling Eq.~(\ref{eq:ProofEreg}). For any $\varepsilon>0$ we can find two integers $k$ and $l$ such that 
\begin{equation}
1<\frac{E_{\mathrm{r}}^{\infty}(\rho^{\otimes k})}{E_{\mathrm{r}}^{\infty}(\sigma^{\otimes l})}<1+\varepsilon.
\end{equation}
Following the same reasoning as above, it is possible to convert the state $\rho^{\otimes k}$ into $\sigma^{\otimes l}$ with rate at least one, which means that $\rho$ can be converted into $\sigma$ with rate at least $l/k$. Recall that $E_{\mathrm{r}}^{\infty}$ is additive on the same state, i.e., $E_{\mathrm{r}}^{\infty}(\rho^{\otimes k})=kE_{\mathrm{r}}^{\infty}(\rho)$ and similar for $\sigma$. It follows that 
\begin{equation}
\frac{l}{k}>\frac{E_{\mathrm{r}}^{\infty}(\rho)}{(1+\varepsilon)E_{\mathrm{r}}^{\infty}(\sigma)},
\end{equation}
which means that $E_{\mathrm{r}}^{\infty}(\rho)/E_{\mathrm{r}}^{\infty}(\sigma)$ is a feasible transformation rate in this setting.

In Ref.~\cite{Lami2023}, an example of a state with different distillable entanglement and entanglement cost with non-entangling operations was presented.
The state is $\rho = \frac{1}{6} \sum_{i,j=1}^3 \pqty{\ketbra{ii} - \ket{ii}\!\bra{jj}}$, which is a maximally correlated state, with distillable entanglement $E_d (\rho) = \log{3/2}$ and entanglement cost $E_c (\rho) = 1$.
Let us compare what happens if we have access to an entanglement battery, quantified by relative entropy.
Since the relative entropy of entanglement is additive if one of the states is a maximally correlated state~\cite{rubboli2023new}, the arguments above show that transformations between maximally correlated states are reversible.
This means that we can reversibly convert the state $\rho$ to a singlet at a rate $E_r(\rho) = \log{3/2}$ with the help of an entanglement battery.
This is in contrasts to the non-entangling operations case, where the entanglement cost is strictly higher.

Let us now analyze the transformation rates if the amount of entanglement in the battery is quantified via the regularized relative entropy of entanglement, i.e., Eq.~(\ref*{eq:Asymptotic-E}) is replaced by 
\begin{equation}
E_{\mathrm{r}}^{\infty}(\tilde{\tau})\geq E_{\mathrm{r}}^{\infty}(\tau).
\end{equation}
Since it is not known if the regularized relative entropy of entanglement is fully additive, this prevents a straightforward application of our Theorems.
Nevertheless, we will see in the following that we obtain a nontrivial theory of entanglement manipulations, with bounded asymptotic entanglement distillation rates. Similar to the non-regularized version, note that $E_{\mathrm{r}}^{\infty}$ is subadditive, i.e.,  
\begin{equation}
E_{\mathrm{r}}^{\infty}(\rho\otimes\tau)\leq E_{\mathrm{r}}^{\infty}(\rho)+E_{\mathrm{r}}^{\infty}(\tau). \label{eq:ERegSubadditivity}
\end{equation}
Moreover, we can show that $E_{\mathrm{r}}^{\infty}$ fulfills Proposition~\ref{prop:ER}. More precisely, for any two states $\rho$ and $\sigma$ in a Hilbert space of dimension $d_{AB}$ with $\frac{1}{2}||\rho-\sigma||_{1}\leq\varepsilon$
and any $\tau$, we have 
\begin{equation}
\left|E_{\mathrm{r}}^{\infty}\left(\rho\otimes\tau\right)-E_{\mathrm{r}}^{\infty}\left(\sigma\otimes\tau\right)\right|\leq\varepsilon\log_{2}d_{AB}+(1+\varepsilon)h\left(\frac{\varepsilon}{1+\varepsilon}\right). \label{eq:ProofRegularized}
\end{equation}
For this, we can use similar arguments as in the proof of Corollary 8 in~\cite{Winter2016}. As can be seen by inspection, the following equality holds:
\begin{align}
 & E_{\mathrm{r}}\left(\rho^{\otimes n}\otimes\tau^{\otimes n}\right)-E_{\mathrm{r}}\left(\sigma^{\otimes n}\otimes\tau^{\otimes n}\right)\\
 & =\sum_{t=1}^{n}E_{\mathrm{r}}\left(\rho^{\otimes t}\otimes\tau^{\otimes t}\otimes\sigma^{\otimes n-t}\otimes\tau^{\otimes n-t}\right)\nonumber \\
 & \quad\quad\quad -E_{\mathrm{r}}\left(\rho^{\otimes t-1}\otimes\tau^{\otimes t-1}\otimes\sigma^{\otimes n-t+1}\otimes\tau^{\otimes n-t+1}\right)\nonumber \\
 & =\sum_{t=1}^{n}E_{\mathrm{r}}\left(\rho\otimes\tau\otimes\Omega_{t}\right)-E_{\mathrm{r}}\left(\sigma\otimes\tau\otimes\Omega_{t}\right)\nonumber 
\end{align}
with $\Omega_{t}=(\rho\otimes\tau)^{\otimes t-1}\otimes(\sigma\otimes\tau)^{\otimes n-t}$. Using triangle inequality, we obtain:
\begin{align}
 & \left|E_{\mathrm{r}}\left(\rho^{\otimes n}\otimes\tau^{\otimes n}\right)-E_{\mathrm{r}}\left(\sigma^{\otimes n}\otimes\tau^{\otimes n}\right)\right|\\
 & \leq\sum_{t=1}^{n}\left|E_{\mathrm{r}}\left(\rho\otimes\tau\otimes\Omega_{t}\right)-E_{\mathrm{r}}\left(\sigma\otimes\tau\otimes\Omega_{t}\right)\right|.\nonumber 
\end{align}
Using Proposition~\ref{prop:ER}, we further obtain
\begin{align}
 & \left|E_{\mathrm{r}}\left(\rho\otimes\tau\otimes\Omega_{t}\right)-E_{\mathrm{r}}\left(\sigma\otimes\tau\otimes\Omega_{t}\right)\right|\nonumber \\
 & \leq\varepsilon\log_{2}d_{AB}+(1+\varepsilon)h\left(\frac{\varepsilon}{1+\varepsilon}\right)
\end{align}
for any $t$.
Collecting the above results gives us
\begin{align}
\left|E_{\mathrm{r}}\left(\rho^{\otimes n}\otimes\tau^{\otimes n}\right)-E_{\mathrm{r}}\left(\sigma^{\otimes n}\otimes\tau^{\otimes n}\right)\right| & \leq n\varepsilon\log_{2}d_{AB}\\
 & +n(1+\varepsilon)h\left(\frac{\varepsilon}{1+\varepsilon}\right).\nonumber 
\end{align}
From this expression we directly obtain 
\begin{align}
\left|\frac{1}{n}E_{\mathrm{r}}\left(\rho^{\otimes n}\otimes\tau^{\otimes n}\right)-\frac{1}{n}E_{\mathrm{r}}\left(\sigma^{\otimes n}\otimes\tau^{\otimes n}\right)\right| & \leq\varepsilon\log_{2}d_{AB}\\
 & +(1+\varepsilon)h\left(\frac{\varepsilon}{1+\varepsilon}\right),\nonumber 
\end{align}
which implies the claimed inequality~(\ref{eq:ProofRegularized}) by taking the limit $n \rightarrow \infty$.

Now, we can obtain statements analogous to Eq.~(\ref{eq:ProofEr-5}) for the transition $\rho \rightarrow \phi^+$:
\begin{align}\label{distil1}
m+E_{\mathrm{r}}^{\infty}\left(\tilde{\tau}\right) & =E_{\mathrm{r}}^{\infty}\left(\ket{\phi^{+}}\!\bra{\phi^{+}}^{\otimes m}\otimes\tilde{\tau}\right)\\
 & \leq E_{\mathrm{r}}^{\infty}\left(\mu^{S_{1}\ldots S_{m}}\otimes\tilde{\tau}\right)+\frac{\varepsilon}{2}m+\left[1+\frac{\varepsilon}{2}\right]h\left(\frac{\frac{\varepsilon}{2}}{1+\frac{\varepsilon}{2}}\right)\nonumber\nonumber \\
 & \leq E_{\mathrm{r}}^{\infty}\left(\rho^{\otimes n}\otimes\tau\right)+\frac{\varepsilon}{2}m+\left[1+\frac{\varepsilon}{2}\right]h\left(\frac{\frac{\varepsilon}{2}}{1+\frac{\varepsilon}{2}}\right)\nonumber\nonumber \\
 & \leq nE_{\mathrm{r}}^{\infty}\left(\rho\right)+E_{\mathrm{r}}^{\infty}\left(\tau\right)+\frac{\varepsilon}{2}m+\left[1+\frac{\varepsilon}{2}\right]h\left(\frac{\frac{\varepsilon}{2}}{1+\frac{\varepsilon}{2}}\right).\nonumber 
\end{align}
Using the same arguments as below Eq.~(\ref{eq:ProofEr-5}), we conclude that the asymptotic entanglement distillation rate is upper bounded by $E_{\mathrm{r}}^{\infty}(\rho)$. Moreover, for any two entangled states $\rho$ and $\sigma$ a transformation at rate $E_{\mathrm{r}}^{\infty}(\rho)/E_{\mathrm{r}}^{\infty}(\sigma)$ can be achieved by using the protocol described below Eq.~(\ref{eq:Ereg}). 

We further notice that if the regularized relative entropy of entanglement is additive, i.e., if the inequality~(\ref{eq:ERegSubadditivity}) is an equality for all states $\rho$ and $\tau$, then by asymptotic continuity of $E_{\mathrm{r}}^{\infty}$~\cite{christandl2006structure,Winter2016} we can apply Theorem~2. It follows that the optimal transformation rate for any pair of states in this setting is given by $R(\rho\rightarrow\sigma)=E_{\mathrm{r}}^{\infty}(\rho)/E_{\mathrm{r}}^{\infty}(\sigma)$, which interestingly coincides with the rate conjectured in~\cite{Brandao2008}. 

\subsubsection{Logarithmic negativity}
We will now quantify entanglement using logarithmic negativity~\cite{PhysRevA.58.883,PhysRevA.65.032314,PhysRevLett.95.090503}
\begin{equation}
E_{\mathrm{n}}(\rho)=\log_{2}\left\Vert \rho^{T_{A}}\right\Vert _{1},
\end{equation}
and investigate transformations with entanglement battery, replacing condition (\ref*{eq:LOCCEB-2}) by 
\begin{equation}
E_{\mathrm{n}}(\tilde{\tau})\geq E_{\mathrm{n}}(\tau).
\end{equation}

Recalling that logarithmic negativity is additive and monotonic under LOCC~\cite{PhysRevA.65.032314}, we directly see that Theorem~1 also holds in this setting. Moreover, due to the additivity of logarithmic negativity, we immediately see that the zero-error rates are given by 
\begin{equation}
R_{\mathrm{ze}}(\rho\rightarrow\sigma)=\frac{E_{\mathrm{n}}(\rho)}{E_{\mathrm{n}}(\sigma)}.
\end{equation}

We will now consider asymptotic transformation with an error margin which vanishes in the asymptotic limit, as defined in Eqs.~(\ref*{eq:Asymptotic}).
Logarithmic negativity is not asymptotically continuous~\cite{Lami2023}, which prevents a direct application of Theorem~2.
However, we will show that it still leads to a nontrivial theory of entanglement manipulation, with bounded singlet distillation rates.
To see this, consider the asymptotic transformation $\rho \rightarrow \phi^+$.
Due to additivity, note that the final state $\mu^{S_{1}\ldots S_{m}}$ fulfills
\begin{equation}
E_{\mathrm{n}}\left(\mu^{S_{1}\ldots S_{m}}\right)\leq nE_{\mathrm{n}}(\rho).
\end{equation}
Assume now that $\mu^{S_{1}\ldots S_{m}}$ is close to $m$ Bell states, i.e., 
\begin{equation}
\left\Vert \mu^{S_{1}\ldots S_{m}}-\ket{\phi^{+}}\!\bra{\phi^{+}}^{\otimes m}\right\Vert _{1}<\varepsilon.
\end{equation}
In the next step, we use the following continuity bound~\cite{Paulsen_2003}
\begin{equation}
\left\Vert \rho^{T_{A}}\right\Vert _{1}-\left\Vert \sigma^{T_{A}}\right\Vert _{1}\leq d_A \left\Vert \rho-\sigma\right\Vert _{1}.
\end{equation}
This implies that 
\begin{equation}
\left\Vert \mu^{T_{A}}\right\Vert _{1}\geq\left\Vert \ket{\phi^{+}}\!\bra{\phi^{+}}^{T_{A}}\right\Vert _{1}^{m}-2^{m}\varepsilon,
\end{equation}
where $\mu^{T_A}$ denotes the partial transpose of the state $\mu^{S_1 \ldots S_m}$.

Collecting the above results and recalling that $||\ket{\phi^{+}}\!\bra{\phi^{+}}^{T_{A}}||_{1}=2$, we obtain 
\begin{align}
nE_{\mathrm{n}}(\rho) & \geq E_{\mathrm{n}}\left(\mu^{S_{1}\ldots S_{m}}\right)=\log_{2}\left(\left\Vert \mu^{T_{A}}\right\Vert _{1}\right)\\
 & \geq\log_{2}\left(\left\Vert \ket{\phi^{+}}\!\bra{\phi^{+}}^{T_{A}}\right\Vert _{1}^{m}-2^{m}\varepsilon\right)\nonumber \\
 & =m+\log_{2}(1-\varepsilon). \nonumber
\end{align}
This inequality can also be expressed as 
\begin{equation}
\frac{m}{n}\leq E_{\mathrm{n}}(\rho)-\frac{1}{n}\log_{2}(1-\varepsilon).
\end{equation}
Since we can choose arbitrary $\varepsilon > 0$, this result means that the asymptotic transformation rate for the conversion $\rho \rightarrow \phi^+$ is upper bounded by $E_\mathrm n (\rho)$. It is further clear that $E_\mathrm n (\rho)$ is a feasible rate for the transformation $\rho \rightarrow \phi^+$, as can be seen using the same techniques as in the proof of Theorem~2. This gives an operational meaning to logarithmic negativity as the optimal rate of distilling singlets in the presence of a resource battery.

\subsubsection{Geometric entanglement}

We will now show that not all entanglement quantifiers lead to a meaningful theory for entanglement manipulation. This can be demonstrated in particular for the geometric entanglement, defined as~\cite{PhysRevA.68.042307,Streltsov_2010}
\begin{equation}
E_{\mathrm{g}}(\rho)=1-\max_{\sigma\in\mathcal{S}}F(\rho,\sigma)
\end{equation}
with fidelity $F(\rho,\sigma)= \pqty{ \mathrm{Tr}\sqrt{\sqrt{\rho}\sigma\sqrt{\rho}}}^{2}$ and $\mathcal S$ is the set of separable states. We now consider asymptotic transformation rates as defined in Eq.~(\ref*{eq:Asymptotic}), where the condition~(\ref*{eq:Asymptotic-E}) is replaced by 
\begin{equation}
E_{\mathrm{g}}(\tilde \tau^{C})\geq E_{\mathrm{g}}(\tau^{C}),\label{eq:Asymptotic-Ed}
\end{equation}
i.e., the amount of entanglement in the battery is constrained by the geometric entanglement. 

As we will now show, in this setting it is possible to convert $n$ Bell states $\ket{\phi^+}$ into $rn$ copies of $\ket{\phi^+}$ for any $r$ with arbitrary accuracy. For this, consider a pure state of the form 
\begin{equation}
\ket{\psi}=\frac{1}{\sqrt{2}}\ket{00}+\frac{1}{\sqrt{2(d-1)}}\sum_{i=1}^{d-1}\ket{ii},
\end{equation}
where $d$ is the local dimension of Alice and Bob. The geometric entanglement of this state is given by $E_{\mathrm{g}}(\ket{\psi})=1/2$, which is the same amount as in the Bell state $\ket{\phi^{+}}=(\ket{00}+\ket{11})/\sqrt{2}$. At the same time, the entanglement entropy of $\ket{\psi}$ is given by 
\begin{equation}
S(\psi^{A})=1+\frac{1}{2}\log_{2}(d-1),
\end{equation}
which is unbounded for large $d$. This also means that the distillable entanglement of $\ket{\psi}$ is unbounded~\cite{PhysRevA.53.2046}. 

Assume now that Alice and Bob share an initial state of the form $\ket{\phi^{+}}^{\otimes n}\otimes\ket{\psi}^{\otimes n}$, where the battery is in the state $\ket{\psi}^{\otimes n}$.
If Alice and Bob permute the primary system and the battery, they obtain the state $\ket{\psi}^{\otimes n}\otimes\ket{\phi^{+}}^{\otimes n}$, where the final state of the battery is $\ket{\phi^{+}}^{\otimes n}$.
Noting that 
\begin{equation}
E_{\mathrm{g}}(\ket{\phi^{+}}^{\otimes n})=E_{\mathrm{g}}(\ket{\psi}^{\otimes n})=1-\frac{1}{2^{n}}
\end{equation}
we conclude that it is possible to convert $n$ Bell states into $n$ copies of the state $\ket{\psi}$, while preserving the geometric entanglement of the battery.
Since we can perform another LOCC protocol to distill singlets, it follows that in this setting $\ket{\phi^+}^{\otimes n}$ can be converted into $\ket{\phi^+}^{\otimes n S(\psi^A)}$ with arbitrary accuracy for large enough $n$.
Moreover, $S(\psi^A)$ can be made arbitrarily large by appropriately choosing the local dimension $d$.

In summary, we have shown that choosing the geometric entanglement for the setting considered in this article leads to a trivial theory of entanglement.
Note that this is distinct from the phenomenon of self-dilution observed with logarithmic negativity, since in that case, the distillation rates are still bounded.
Here, the divergence of distillation rates is caused by the lack of sensitivity in geometric entanglement to distinguish $\ket{\psi}$ and $\ket{\phi^+}$.

\subsection{Connection to catalysis}

The notion of a resource battery is closely related to catalysis~\cite{Datta_2023,lipka2023catalysis}, which is studied extensively in resource theories, entanglement being a special case thereof. Resource theories are characterized by free operations and states, that describe the abilities that an agent has unlimited access to. In entanglement theory, the standard set of free operations often correspond to LOCC while free states being separable states; in thermodynamics, thermal operations form the basic set of free operations, with Gibbs thermal states being free states.

Given any resource theory, we say that $\rho$ can be transformed to $\sigma$ with (exact) catalysis if we can find a catalyst state $\tau$, and a free operation $\Lambda$ such that
\begin{align}
    \Lambda \pqty{\rho^S \otimes \tau^C} &= \sigma^S \otimes \tau^C.
\end{align}
Comparing this to our conditions on e.g. an entanglement battery Eqs.~(\ref*{eq:LOCCEB-1}) and (\ref*{eq:LOCCEB-2}), we see that catalysis is a stricter condition: in catalysis, the state of the catalyst must be preserved, while Eq.~(\ref*{eq:LOCCEB-2}) only requires the entanglement in the battery is preserved.

We can show that we can recover the catalytic condition if we impose additional restrictions on the battery. Taking entanglement theory as an example,
let $\Bqty{E_i}$ be a complete set of monotones for LOCC, i.e.\ there is an LOCC protocol transforming $\rho$ to $\sigma$ if and only if $E_i (\rho) \geq E_i (\sigma)$ for all $i$.
Let us now impose the following condition on the battery-assisted transformations: the battery entanglement must not decrease for all of these $E_i$'s, namely $E_i(\tilde{\tau}) \geq E_i (\tau)$ for all $i$.
In this case, the completeness of $\Bqty{E_i}$ implies that there exists an LOCC protocol that transform $\tilde{\tau}$ back to $\tau$.
Thus we can always post-process the battery into its initial state, and we recover the catalytic condition.
This implies that by imposing only a single entanglement measure, we are ignoring all of the other types of entanglement that are present in the battery. Complete sets of monotones for entangled state transformations have been investigated in~\cite{PhysRevLett.83.436,PhysRevLett.83.1046,PhysRevA.72.022323,PhysRevLett.130.240204,Takagi_2019}, while in thermodynamics, the complete set of monotones have also been found for smaller sets of state transformations~\cite{horodecki2013fundamental}.

There is another connection to catalysis.
Let us examine what happens to the battery state in a transformation cycle $\rho \to \sigma \to \rho$ when we perform the swapping protocol.
Since $R(\rho \to \sigma) R(\sigma \to \rho) = 1$, for any $\delta > 0$, there is $m, n, s$ s.t. $\rho^{\otimes m} \to \sigma^{\otimes n} \to \rho^{\otimes s}$ and $s/m \geq 1 - \delta$.
With the swapping protocol, initially we prepare the battery in the state $\tau = \sigma^{\otimes n} \otimes \rho^{\otimes s}$.
After the first stage of the transformation, the battery is in the state $\tau' = \rho^{\otimes m} \otimes \rho^{\otimes s}$ and after the second stage, the state is $\tau'' = \sigma^{\otimes n} \otimes \rho^{\otimes m}$.
Notice that although the state of the battery after the first stage of the cycle $\tau'$ can be far from its initial state $\tau$, the final state of the battery $\tau''$ is near its initial state $\tau$.
This shows that the battery behaves almost catalytically in the whole cycle, although it might have large fluctuations in each individual stage.

\subsection{Correlations between main system and battery}

We will now consider a more general setting where the battery might
become correlated with the main system. We say that a state $\rho$
can be transformed into another state $\sigma$ in this setting if
for any $\varepsilon>0$, there exists an LOCC protocol $\Lambda$
and a state of the battery $\tau$ such that \begin{subequations} \label{eq:CorrelatedBattery}
\begin{align}
\Lambda\left(\rho^{AB}\otimes\tau^{A'B'}\right) & =\mu^{ABA'B'},\\
\left\Vert \mu^{AB}-\sigma^{AB}\right\Vert _{1} & <\varepsilon,\\
E(\mu^{A'B'}) & \geq E(\tau^{A'B'}).
\end{align}
\end{subequations}

To obtain nontrivial state transformations, we consider entanglement measures
$E$ having the following properties:
\begin{enumerate}
\item Monotonicity under LOCC: 
\begin{equation}
E(\Lambda[\rho])\leq E(\rho)
\end{equation}
for any LOCC protocol $\Lambda$.
\item Superadditivity: 
\begin{equation}
E^{AA'|BB'}(\rho^{AA'BB'})\geq E(\rho^{AB})+E(\rho^{A'B'}).
\end{equation}
\item Additivity on product states: 
\begin{equation}
E^{AA'|BB'}(\rho^{AB}\otimes\sigma^{A'B'})=E(\rho^{AB})+E(\sigma^{A'B'}).
\end{equation}
\item Continuity.
\end{enumerate}
Note that asymptotic continuity is not required in the following. Examples of entanglement measures fulfilling these properties are the squashed entanglement~\cite{10.1063/1.1643788} and the conditional entanglement of mutual information~\cite{PhysRevLett.101.140501,yang2007conditional}. To see that conditional entanglement of mutual information fulfills superadditivity (property 2), consider an extension $\tau^{AA'A''BB'B''}$ of the state $\rho^{AA'BB'}$. It holds that
\begin{align}
 & I(AA'A'':BB'B'')-I(A'':B'')\\
 & =I(AA'A'':BB'B'')-I(A'A'':B'B'')\nonumber \\
 & \quad +I(A'A'':B'B'')-I(A'':B''),\nonumber 
\end{align}
which directly implies that $E(\rho^{AA'BB'})\geq E(\rho^{AB})+E(\rho^{A'B'})$.

Continuity together with Eqs.~(\ref{eq:CorrelatedBattery}) implies that for
any $\delta>0$, there is an LOCC protocol $\Lambda$ and a state of
the battery $\tau$ such that \begin{subequations}
\begin{align}
\Lambda\left(\rho^{AB}\otimes\tau^{A'B'}\right) & =\mu^{ABA'B'},\\
\left|E\left(\sigma^{AB}\right)-E\left(\mu^{AB}\right)\right| & <\delta,\\
E(\mu^{A'B'}) & \geq E(\tau^{A'B'}).
\end{align}
\end{subequations}

We will now show that state transformations in this setting are
fully characterized by the amount of entanglement $E$, i.e., a state
$\rho^{AB}$ can be converted into $\sigma^{AB}$ if and only if
\begin{equation}
E(\rho^{AB})\geq E(\sigma^{AB}).\label{eq:CorrelatedBatteryE}
\end{equation}
To prove this, we will first show that the amount of entanglement
in the main system $AB$ cannot increase in this procedure. If $\rho^{AB}$
can be converted into $\sigma^{AB}$, then
\begin{align}
E\left(\sigma^{AB}\right) & \leq E\left(\mu^{AB}\right)+\delta\leq E^{AA'|BB'}\left(\mu^{ABA'B'}\right)-E\left(\mu^{A'B'}\right)+\delta\nonumber \\
 & \leq E^{AA'|BB'}\left(\rho^{AB}\otimes\tau^{A'B'}\right)-E\left(\mu^{A'B'}\right)+\delta \nonumber \\
 & =E\left(\rho^{AB}\right)+E\left(\tau^{A'B'}\right)-E\left(\mu^{A'B'}\right)+\delta \nonumber \\
 & \leq E\left(\rho^{AB}\right)+\delta. 
\end{align}
Since $\delta>0$ can be chosen arbitrarily, we obtain Eq.~(\ref{eq:CorrelatedBatteryE})
as claimed. The converse can be shown using the same protocol as in
the proof of Theorem~2, i.e., initializing the battery in the desired
target state, and then permuting the state of the main system and
the battery.

We will now show that for bipartite pure states $\ket{\psi}^{AB}$
and $\ket{\phi}^{AB}$, the transitions are fully characterized by
the entanglement entropy, i.e., the transformation $\ket{\psi}^{AB}\rightarrow\ket{\phi}^{AB}$
is possible if and only if 
\begin{equation}
S(\psi^{A})\geq S(\phi^{A}).\label{eq:PureStateCatalysis}
\end{equation}
Interestingly, this result does not depend on the entanglement measure used, as long
as the measure fulfills the properties 1-4 mentioned above.
Note that this generalizes the result of Ref.~\cite{alhambra2019entanglement} by allowing general mixed states as a battery.

For proving this statement, recall that Eq.~(\ref{eq:PureStateCatalysis}) completely characterizes
transformations between pure states in a catalytic setting~\cite{PhysRevLett.127.150503}, which is a more restrictive setting than allowing for a correlated battery.
In more detail, we say that $\rho$ can be converted into $\sigma$ via LOCC with correlated catalyst if for any $\varepsilon>0$ there
exists an LOCC protocol $\Lambda$ and a state $\tau$
such that~\cite{PhysRevLett.127.150503,ganardi2023catalytic,lami2023catalysis}
\begin{subequations}
\begin{align}
\Lambda\left(\rho^{AB}\otimes\tau^{A'B'}\right) & =\mu^{ABA'B'},\\
\left\Vert \mu^{AB}-\sigma^{AB}\right\Vert _{1} & <\varepsilon,\\
\mu^{A'B'} & =\tau^{A'B'}.
\end{align}
\end{subequations} 
As follows from the results in~\cite{PhysRevLett.127.150503}, a pure state $\ket{\psi}^{AB}$
can be converted into another pure state $\ket{\phi}^{AB}$ in this
setting if and only if Eq.~(\ref{eq:PureStateCatalysis}) is fulfilled.
Note that any entanglement measure $E$ which fulfills conditions 1-4 mentioned above does not increase in this setting~\cite{PhysRevLett.127.150503}. In more detail, if $\rho^{AB}$ can
be converted into $\sigma^{AB}$ via LOCC with correlated catalysis,
it holds that $E(\sigma^{AB})\leq E(\rho^{AB})$. 

Summarizing these
arguments, any entanglement measure $E$ which fulfills the properties 1-4 mentioned above must have the same monotonicity on pure states as the entanglement
entropy:
\begin{equation}
E\left(\ket{\psi}^{AB}\right)\geq E\left(\ket{\phi}^{AB}\right)\iff S\left(\psi^{A}\right)\geq S\left(\phi^{A}\right).
\end{equation}
 This shows that Eq.~(\ref{eq:PureStateCatalysis}) is necessary and sufficient for a transition $\ket{\psi}^{AB}\rightarrow\ket{\phi}^{AB}$ via LOCC with entanglement battery, as defined in Eqs.~(\ref{eq:CorrelatedBattery}).

We will now consider asymptotic state transformations in this setting.
We say that $\rho$ can be converted into $\sigma$ at rate $r$ if for any $\varepsilon,\delta > 0$ there exist integers $m$, $n$, an LOCC protocol $\Lambda$, and a state of the battery $\tau$ such that 
\begin{subequations}
\begin{align}
\Lambda\left(\rho^{\otimes n}\otimes\tau^{C}\right) & =\mu^{S_{1}\ldots S_{m}C},\\
\left\Vert \mu^{S_{1}\ldots S_{m}}-\sigma^{\otimes m}\right\Vert _{1} & <\varepsilon,\\
E(\mu^{C}) & \geq E(\tau^{C}),\\
\frac{m}{n}+\delta & >r.
\end{align}
\end{subequations}
In the following, $R(\rho \rightarrow \sigma)$ denotes the supremum over all feasible rates $r$ in this setting.

We will now show that for any entanglement measure satisfying the properties 1-4 mentioned above, the asymptotic rate is given by 
\begin{equation}
R(\rho\rightarrow\sigma)=\frac{E(\rho)}{E(\sigma)}. \label{eq:SuperAdditiveMeasureAsymptoticRate}
\end{equation}
We will first show that the transformation rate is upper bounded as 
\begin{equation}
R(\rho\rightarrow\sigma) \leq \frac{E(\rho)}{E(\sigma)}. \label{eq:UpperBound}
\end{equation}
For this, we consider the following relaxation of the asymptotic transformation task (a similar technique has been used in~\cite{ganardi2023catalytic}). Instead of establishing $m$ copies of $\sigma$ from $n$ copies of $\rho$ with a battery $C$, our goal is to establish a state $\mu^{S_1\ldots S_mC}$, with each of the subsystems $S_i$ having almost the same amount of entanglement as $\sigma$. In more detail, we require that for any $\varepsilon,\delta > 0$ there are integers $m$ and $n$, an LOCC protocol $\Lambda$ and a state of the battery $\tau^C$ such that for all $i \leq m$,
\begin{subequations} \label{eq:Asymptotic-2}
\begin{align}
\Lambda\left(\rho^{\otimes n}\otimes\tau^{C}\right) & =\mu^{S_{1}\ldots S_{m}C},\\
\left|E(\mu^{S_{i}})-E(\sigma)\right| & <\varepsilon,\\
E(\mu^{C}) & \geq E(\tau^{C}),\\
\frac{m}{n} & >r-\delta.
\end{align}
\end{subequations}
Here, $r$ is a feasible transformation rate in this process, and the supremum over all feasible rates will be denoted by $R'(\rho \rightarrow \sigma)$. Recalling that $E$ is continuous (property 4), we have 
\begin{equation}
    R'(\rho \rightarrow \sigma) \geq R(\rho \rightarrow \sigma). \label{eq:R'}
\end{equation}

Using Eqs.~(\ref{eq:Asymptotic-2}) and the properties 1-4 of $E$ we find
\begin{align}
nE\left(\rho\right)+E\left(\tau\right) & =E\left(\rho^{\otimes n}\otimes\tau^{C}\right)\geq E\left(\mu^{S_{1}\ldots S_{m}C}\right)\\
 & \geq\sum_{i=1}^{m}E\left(\mu^{S_{i}}\right)+E\left(\mu^{C}\right),\nonumber 
\end{align}
which further implies
\begin{equation}
nE\left(\rho\right)\geq\sum_{i=1}^{m}E\left(\mu^{S_{i}}\right).
\end{equation}
Using once again Eqs.~(\ref{eq:Asymptotic-2}) we arrive at 
\begin{equation}
nE(\rho)>m\left[E(\sigma)-\varepsilon\right],
\end{equation}
leading to 
\begin{equation}
\frac{m}{n}<\frac{E(\rho)}{E(\sigma)-\varepsilon}.
\end{equation}
Applying Eqs.~(\ref{eq:Asymptotic-2}) one more time we obtain
\begin{equation}
r<\frac{E(\rho)}{E(\sigma)-\varepsilon}+\delta.
\end{equation}
Since $\varepsilon > 0$ and $\delta > 0$ can be chosen arbitrarily, we obtain 
\begin{equation}
R'(\rho\rightarrow\sigma)\leq\frac{E(\rho)}{E(\sigma)}.
\end{equation}
Together with Eq.~(\ref{eq:R'}) this completes the proof of Eq.~(\ref{eq:UpperBound}).

To complete the proof of Eq.~(\ref{eq:SuperAdditiveMeasureAsymptoticRate}), we need to show the converse, i.e., a protocol achieving conversion at rate $E(\rho)/E(\sigma)$. This can be done in the same way as in the proof of Theorem~2. In summary, we have proven that the asymptotic rates are given by Eq.~(\ref{eq:SuperAdditiveMeasureAsymptoticRate}) in this setting. This applies to any entanglement quantifier $E$ which fulfills properties 1-4 mentioned above.

\subsection{Self-dilution in thermodynamics}

As mentioned, self-dilution also occurs in thermodynamics if we quantify the battery using $F_{\textrm{max}}$
\begin{align}
    F_{\textrm{max}} (\rho) = \inf \Bqty{\log \lambda \,|\, \rho \leq \lambda \gamma}.
\end{align}
As with entanglement, we look at the protocol where we perform a distillation with a battery and then dilute the result to get the initial state back.
Let us focus on transitions between incoherent states on a single qubit.
In this setting, we can perform a distillation/dilution into $\ketbra{1}$ reversibly, with a rate determined by free energy $F$~\cite{Brandao_2013}.
Furthermore, since we only look at transformations between incoherent states, no additional source of coherence is needed in the dilution process.

Now, let us look at the distillation rate $R(\rho \to \ketbra{1})$ when we use a free energy battery quantified by $F_{\textrm{max}}$.
Note that $F_{\textrm{max}}$ is additive, so the final state $\mu$ satisfies $n F_{\textrm{max}} (\rho) \geq F_{\textrm{max}} (\mu)$.
Furthermore, $F_{\textrm{max}}$ satisfies the asymptotic equipartition property~\cite{tomamichel2015quantum}, so 
\begin{align}
    \lim_{\epsilon \to 0}
    \lim_{m \to \infty}
    \inf_{\norm{\mu - \ketbra{1}^{\otimes m}} \leq \epsilon}
    \frac{1}{m} F_{\textrm{max}} (\mu)
    &=
    F\pqty{\ketbra{1}}.
\end{align}
Now, for any $\epsilon, \delta > 0$, we have $m, n, \mu$ from Eqs.~(\ref*{eq:Asymptotic}) which satisfies
\begin{align}
    \frac{m}{n}
    &\leq \frac{F_{\textrm{max}}(\rho)}{F_{\textrm{max}}(\mu) /m}.
\end{align}
Using the asymptotic equipartition property and the definition of $r$, we get
\begin{align}
    r
    &\leq
    \lim_{\epsilon \to 0}
    \lim_{m \to \infty}
    \sup_{\norm{\mu - \ketbra{1}^{\otimes m}} \leq \epsilon}
    \frac{F_{\textrm{max}}(\rho)}{F_{\textrm{max}}(\mu) /m}
    \\
    &= \frac{F_{\textrm{max}}(\rho)}{F(\ketbra{1})} + \delta. \nonumber
\end{align}
Since this holds for any $\delta > 0$, we must have $R(\rho \to \ketbra{1}) \leq F_{\textrm{max}}(\rho)/F\pqty{\ketbra{1}}$, and in particular the distillation rate is bounded.

Now, consider the following process: Starting with $\rho^{\otimes n}$, we distill $\ketbra{1}^{\otimes n r}$ with the help of a battery, where $r = F_{\textrm{max}}(\rho)/F_{\textrm{max}}(\ketbra{1})$.
With a battery, this transformation can be achieved with zero error.
Then, we perform a dilution procedure $\ketbra{1}^{\otimes n r} \to \rho^{\otimes n r r'}$ \emph{without} battery, with $r' = F(\ketbra{1})/F(\rho)$~\cite{Brandao_2013}.
Then, we have
\begin{align}
    r r'
    &=
    \frac{F_{\textrm{max}}(\rho)}{F(\rho)} 
    \frac{F(\ketbra{1})}{F_{\textrm{max}}(\ketbra{1})},
\end{align}
and we can perform the transformation $\rho^{\otimes n} \to \rho^{\otimes n r r'}$ with a battery. 
We finish by noting that in general we have $F_{\textrm{max}} (\rho) \geq F(\rho)$, and for incoherent states on a single qubit $\rho = (1-p) \ketbra{0} + p \ketbra{1}$, we have equality only when $p = 0, 1$ or when $\rho = \gamma$.

\bibliography{literature}

\end{document}